\pdfoutput=1
\documentclass{article}
\usepackage{paper}
\usepackage{subcaption}
\usepackage{multirow}
\newcounter{num}
\makeatletter
\newcommand{\TODO}[1]{\textcolor{red}{[TODO\@ifnotempty{#1}{: #1}]}}
\newcommand{\sandeep}[1]{\textcolor{orange}{[sandeep\@ifnotempty{#1}{: #1}]}}
\newcommand{\piotr}[1]{\textcolor{blue}{[piotr\@ifnotempty{#1}{: #1}]}}

\makeatother

\usepackage{algorithm}
\usepackage{algpseudocode}
\usepackage{wrapfig}
\usepackage{thmtools}
\usepackage{thm-restate}

\title{Faster Linear Algebra for Distance Matrices}
\author{Piotr Indyk \\ MIT \\ \texttt{indyk@mit.edu}  \and Sandeep Silwal\\ MIT  \\ \texttt{silwal@mit.edu}}
\date{}

\begin{document}
\maketitle
\begin{abstract}
   The distance matrix of a dataset $X$ of $n$ points with respect to a distance function $f$ represents all pairwise distances between points in $X$ induced by $f$. Due to their wide applicability, distance matrices and related families of matrices have been the focus of many recent algorithmic works. We continue this line of research and take a broad view of algorithm design for distance matrices with the goal of designing fast algorithms, which are specifically tailored for distance matrices, for fundamental linear algebraic primitives. Our results include efficient algorithms for computing matrix-vector products for a wide class of distance matrices, such as the $\ell_1$ metric for which we get a linear runtime, as well as an $\Omega(n^2)$ lower bound for any algorithm which computes a matrix-vector product for the $\ell_{\infty}$ case, showing a separation between the $\ell_1$ and the $\ell_{\infty}$ metrics. Our upper bound results, in conjunction with recent works on the matrix-vector query model, have many further downstream applications, including the fastest algorithm for computing a relative error low-rank approximation for the distance matrix induced by $\ell_1$ and $\ell_2^2$ functions and the fastest algorithm for computing an additive error low-rank approximation for the $\ell_2$ metric, in addition to applications for fast matrix multiplication among others. We also give algorithms for constructing distance matrices and show that one can construct an approximate $\ell_2$ distance matrix in time faster than the bound implied by the Johnson-Lindenstrauss lemma.
\end{abstract}

\section{Introduction}

Given a set of $n$ points $X = \{x_1, \ldots, x_n\}$, the distance matrix of $X$ with respect to a distance function $f$ is defined as the $n \times n$ matrix $A$ satisfying $A_{i,j} = f(x_i, x_j)$. Distances matrices are ubiquitous objects arising in various applications ranging from learning image manifolds \cite{tenenbaum2000global, weinberger2006unsupervised}, signal processing \cite{so2007theory}, biological analysis \cite{holm1993protein}, and non-linear dimensionality reduction \cite{kruskal1964multidimensional, kruskal1978multidimensional, tenenbaum2000global, cox2008multidimensional}, to name a few\footnote{We refer the reader to the survey \cite{dokmanic2015euclidean} for a more thorough discussion of applications of distance matrices.}. Unfortunately, explicitly computing and storing $A$ requires at least $\Omega(n^2)$ time and space. Such complexities are prohibitive for scaling to large datasets.

A silver lining is that in many settings, the matrix $A$ is not explicitly required. Indeed in many applications, it suffices to compute some underlying function or property of $A$, such as the eigenvalues and eigenvectors of $A$ or a low-rank approximation of $A$. Thus an algorithm designer can hope to use the special geometric structure encoded by $A$ to design faster algorithms tailored for such tasks.

Therefore, it is not surprising that many recent works explicitly take advantage of the underlying geometric structure of distance matrices, and other related families of matrices, to design fast algorithms (see Section \ref{sec:related_works} for a thorough discussion of prior works). In this work, we continue this line of research and take a broad view of algorithm design for distance matrices. Our main motivating question is the following:
\begin{quote}
     \textit{Can we design algorithms for fundamental linear algebraic primitives which are specifically tailored for distance matrices and related families of matrices?}
\end{quote}

We make progress towards the motivating question by studying three of the most fundamental primitives in algorithmic linear algebra. Specifically:
\begin{enumerate}
    \item We study upper and lower bounds for computing matrix-vector products for a wide array of distance matrices,
    \item We give algorithms for multiplying distance matrices faster than general matrices,  and,
    \item We give fast algorithms for constructing distance matrices.
\end{enumerate}

\subsection{Our Results}
We now describe our contributions in more detail.  
\begin{quote}
    \hspace{-2mm} \textit{1. We study upper and lower bounds for constructing matrix-vector queries for a wide array of distance matrices.}
\end{quote}
%First we state upper bounds on fast algorithms for matrix-vector queries. 
A matrix-vector query algorithm accepts a vector $z$ as input and outputs the vector $Az$. There is substantial motivation for studying such queries. Indeed, there is now a rich literature for fundamental linear algebra algorithms which are in the ``matrix free" or ``implicit" model. These algorithms only assume access to the underlying matrix via matrix-vector queries. Some well known algorithms in this model include the power method for computing eigenvalues and the conjugate gradient descent method for solving a system of linear equations. For many fundamental functions of $A$, nearly optimal bounds in terms of the number of queries have been achieved  \cite{musco2015randomized, braverman2020gradient, bakshi2022low}. Furthermore, having access to matrix-vector queries also allows the simulation of any randomized sketching algorithm, a well studied algorithmic paradigm in its own right \cite{woodruff2014sketching}. This is because randomized sketching algorithms operate on the matrix $\Pi A$ or $A \Pi$ where $\Pi$ is a suitably chosen random matrix, such as a Gaussian matrix. Typically, $\Pi$ is chosen so that the sketches $\Pi A$ or $A \Pi$ have significantly smaller row or column dimension compared to $A$. If $A$ is symmetric, we can easily acquire both types of matrix sketches via a small number of matrix-vector queries.

Therefore, creating efficient versions of matrix-vector queries for distance matrices automatically lends itself to many further downstream applications. We remark that our algorithms can access to the set of input points but \emph{do not} explicitly create the distance matrix. A canonical example of our upper bound results is the construction of matrix-vector queries for the function $f(x,y) = \|x-y\|_p^p$.

\begin{theorem}
 Let $p \ge 1$ be an integer. Suppose we are given a dataset of $n$ points $X = \{x_1, \ldots, x_n \} \subset \R^d$. $X$ implicitly defines the matrix $A_{i,j} = \|x_i - x_j\|_p^p$. Given a query $z \in \R^n$, we can compute $Az$ exactly in time $O(ndp)$. If $p$ is odd, we also require $O(nd \log n)$ preprocessing time.
\end{theorem}

We give similar guarantees for a wide array of functions $f$ and we refer the reader to Table \ref{tab:results} which summarizes our matrix-vector query upper bound results. Note that some of the functions $f$ we study in Table \ref{tab:results} do not necessarily induce a metric in the strict mathematical sense (for example the function $f(x,y) = \|x-y\|_2^2$ does not satisfy the triangle inequality). Nevertheless, we still refer to such functions under the broad umbrella term of ``distance functions" for ease of notation. We always explicitly state the function $f$ we are referring to.

Crucially, most of our bounds have a linear dependency on $n$ which allows for scalable computation as the size of the dataset $X$ grows. Our upper bounds are optimal in many cases, see Theorem \ref{thm:lb_for_ub}.

\begin{table}[!ht]
\centering
{\renewcommand{\arraystretch}{1.3}
\begin{tabular}{c|c|c|c|c}
Function                     & $f(x,y)$                                                                          & Preprocessing & Query Time     & Reference                                \\ \hline
$\ell_p^p$ for $p$ even      & $\|x-y\|_p^p$                                                                     & $-$                & $O(ndp)$       & Thms. \ref{thm:l2_squared_upper_bound} / \ref{thm:even_p}      \\
$\ell_p^p$ for $p$ odd       & $\|x-y\|_p^p$                                                                     & $O(nd \log n)$     & $O(nd p)$ & Thms. \ref{thm:l1_upper_bound} / \ref{thm:odd_p}       \\
Mixed $\ell_{\infty}$        & $\max_{i,j} |x_i - y_j|$                                                          & $O(nd \log n)$     & $O(n^2)$       & Thm. \ref{thm:mixed}        \\
Mahalanobis Distance$^2$ & $x^TMy$                                                                           & $O(nd^2)$          & $O(nd)$        & Thm. \ref{thm:maha}         \\
Polynomial Kernel            & $\langle x,y\rangle^p$                                                                         & $-$                & $O(nd^p)$      & Thm. \ref{thm:poly_kernel} \\
Total Variation Distance     & $\text{TV}(x,y)$                                                                  & $O(nd \log n)$     & $O(nd)$ & Thm. \ref{thm:tv}           \\
KL Divergence                & $\text{D}_{\text{KL}}(x \, \| \, y)$                                              & $-$                & $O(nd)$        & Thm. \ref{thm:kl}           \\
Symmetric Divergence    & $\text{D}_{\text{KL}}(x \, \| \, y)+\text{D}_{\text{KL}}(y \, \| \, x)$ & $-$                & $O(nd)$        & Thm. \ref{thm:all_kl}      \\
Cross Entropy                & $H(x,y)$                                                                          & $-$                & $O(nd)$        & Thm. \ref{thm:all_kl}      \\
Hellinger Distance$^2$   & $\sum_{i=1}^d \sqrt{x(i) y(i)}$                                                   & $-$                & $O(nd)$        & Thm. \ref{thm:hellinger}   
\end{tabular}
}
\caption{A summary of our results for exact matrix-vector queries.}
\label{tab:results}
\end{table}

Combining our upper bound results with optimized matrix-free methods, immediate corollaries of our results include faster algorithms for eigenvalue and singular value computations and low-rank approximations. Low-rank approximation is of special interest as it has been widely studied for distance matrices; for low-rank approximation, our bounds outperform prior results for specific distance functions. For example, for the $\ell_1$ and $\ell_2^2$ case (and in general PSD matrices), \cite{bakshi2020robust} showed that a rank-$k$ approximation can be found in time $O(ndk/\eps + nk^{w-1}/\eps^{w-1})$. This bound has extra $\text{poly}(1/\eps)$ overhead compared to our bound stated in Table \ref{tab:results2}. The work of \cite{indyk2019sample} has a worse $\text{poly}(k, 1/\eps)$ overhead for an additive error approximation for the $\ell_2$ case. See Section \ref{sec:related_works} for further discussion of prior works. The downstream applications of matrix-vector queries are summarized in Table \ref{tab:results2}.

We also study fundamental limits for any upper bound algorithms. In particular, we show that \emph{no algorithm} can compute a matrix-vector query for general inputs for the $\ell_{\infty}$ metric in subquadratic time, assuming a standard complexity-theory assumption called the {\em Strong Exponential Time Hypothesis (SETH)}~\cite{impagliazzo2001complexity,impagliazzo2001problems}. % \piotr{Lower bound on $d$}.
\begin{theorem}\label{thm:l_infinity_lower_bound}
 For any $\alpha>0$ and $d=\omega(\log n)$, any algorithm for exactly computing $Az$ for any input $z$, where $A$ is the $\ell_{\infty}$ distance matrix, requires $\Omega(n^{2-\alpha})$ time (assuming SETH).
\end{theorem}

This shows a separation between the functions listed in Table \ref{tab:results} and the $\ell_{\infty}$ metric. Surprisingly, we can create queries for the \emph{approximate} matrix-vector query in substantially faster time.

\begin{theorem}
 Suppose $X \subseteq \{0,1, \ldots, O(1)\}^d$.
 We can compute $By$ in time $O(n \cdot d^{O(\sqrt{d} \log (d/\eps))})$ where $\|A-B\|_{\infty} \le \eps$.
\end{theorem}

To put the above result into context, the lower bound of Theorem \ref{thm:l_infinity_lower_bound} holds for points sets in $\{0,1,2\}^d$ in $d \approx \log n$ dimensions. 
%Thus the lower bound can be interpreted as stating that any algorithm for computing exact matrix-vector queries requires $\Omega(n \cdot 2^d)$ time. 
In contrast, if we relax to an approximation guarantee, we can obtain a subquadratic-time algorithm for $d$ up to $\Theta(\log^2(n)/\log \log (n))$. 

%better runtime of roughly $O(n \cdot 2^{\sqrt{d}})$.

Finally, we provide a general understanding of the limits of our upper bound techniques. In Theorem \ref{thm:meta_upperbound}, we show that essentially the only $f$ for which our upper bound techniques apply have a ``linear structure" after a suitable transformation. We refer to Appendix Section \ref{sec:meta} for details.

\begin{table}[!ht]
\centering
{\renewcommand{\arraystretch}{1.3}
\begin{tabular}{c|c|c|c}
Problem                                                                                                 & $f(x,y)$                         & Runtime                                                                                                                                                        & Prior Work                                                                           \\ \hline
\begin{tabular}[c]{@{}c@{}}$(1+\eps)$ Relative error rank $k$\\ low-rank approximation\end{tabular}     & $\ell_1, \ell_2^2$               & \begin{tabular}[c]{@{}c@{}}$\tilde{O}\left( \frac{ndk}{\eps^{1/3}} + \frac{nk^{w-1}}{\eps^{(w-1)/3}} \right)$\\ Theorem \ref{thm:opt_lowrank_p}\end{tabular} & \begin{tabular}[c]{@{}c@{}}$O\left( \frac{ndk}{\eps} + \frac{nk^{w-1}}{\eps^{w-1}} \right)$ \\ \cite{bakshi2020robust} \end{tabular}                    \\ \hline
\begin{tabular}[c]{@{}c@{}}Additive error $\eps \|A\|_F$ rank $k$\\ low-rank approximation\end{tabular} & $\ell_2$                         & \begin{tabular}[c]{@{}c@{}}$\tilde{O}\left( \frac{ndk}{\eps^{1/3}} + \frac{nk^{w-1}}{\eps^{(w-1)/3}} \right)$\\ Theorem \ref{thm:approx_low_rank}\end{tabular} & \begin{tabular}[c]{@{}c@{}} $\tilde{O}(nd \cdot \text{poly}(k, 1/\eps))$ \\ \cite{indyk2019sample} \end{tabular}                                        \\ \hline
\begin{tabular}[c]{@{}c@{}}$(1+\eps)$ Relative error rank $k$\\ low-rank approximation\end{tabular}     & Any in Table \ref{tab:results} &  \begin{tabular}[c]{@{}c@{}}$\tilde{O}\left( \frac{Tk}{\eps^{1/3}} + \frac{nk^{w-1}}{\eps^{(w-1)/3}} \right)$    \\ Theorem \ref{thm:low_rank_meta}\end{tabular}                                                                          & \begin{tabular}[c]{@{}c@{}}$\tilde{O}\left( \frac{n^2dk}{\eps^{1/3}} + \frac{nk^{w-1}}{\eps^{(w-1)/3}} \right)$ \\ \cite{bakshi2022low} \end{tabular} \\ \hline
\begin{tabular}[c]{@{}c@{}}$(1\pm \eps)$ Approximation to\\ top $k$ singular values\end{tabular}        & Any in Table \ref{tab:results} & \begin{tabular}[c]{@{}c@{}}$\tilde{O}\left(\frac{Tk}{\eps^{1/2}} + \frac{nk^2}\eps + \frac{k^3}{\eps^{3/2}}\right)$ \\ Theorem \ref{thm:singular_values_system} \end{tabular}                                                                                                      & \begin{tabular}[c]{@{}c@{}} $\tilde{O}\left(\frac{n^2dk}{\eps^{1/2}} + \frac{nk^2}\eps + \frac{k^3}\eps^{3/2}\right)$ \\ \cite{musco2015randomized}   \end{tabular}                       \\ \hline
\begin{tabular}[c]{@{}c@{}}Multiply distance matrix $A$\\ with any $B \in \R^{n \times n}$\end{tabular} & Any in Table \ref{tab:results} & \begin{tabular}[c]{@{}c@{}}$O(Tn)$ \\ Lemma \ref{lem:matrixmult} \end{tabular}                                                                                                                                                    & $O(n^w)$                                                         \\ \hline
\begin{tabular}[c]{@{}c@{}}Multiply two distance \\ matrices $A$ and $B$\end{tabular}                   & $\ell_2^2$                       & \begin{tabular}[c]{@{}c@{}} $O(n^2d^{w-2})$  \\ Lemma \ref{lem:matrixmultltwo} \end{tabular}                                                                                                                                             & $O(n^w)$                                                                            
\end{tabular}
}
\caption{Applications of our matrix-vector query results. $T$ denotes the matrix-vector query time, given in Table \ref{tab:results}. $w \approx 2.37$ is the matrix multiplication constant \cite{alman2021refined}.}
\label{tab:results2}
\end{table}

\begin{quote} 
\hspace{-2mm}\textit{2. We give algorithms for multiplying distance matrices faster than general matrices.}
\end{quote}

% \piotr{Is it really a different bullet, or a part of the previous one?}
Fast matrix-vector queries also automatically imply fast matrix multiplication, which can be reduced to a series of matrix-vector queries. For concreteness, if $f$ is the $\ell_p^p$ function which induces $A$, and $B$ is any $n \times n$ matrix, we can compute $AB$ in time $O(n^2dp)$. This is substantially faster than the general matrix multiplication bound of $n^w \approx n^{2.37}$. We also give an improvement of this result for the case where we are multiplying two distance matrices arising from $\ell_2^2$. See Table \ref{tab:results2} for summary.

\begin{quote}
    \hspace{-2mm}\textit{3. We give fast algorithms for constructing distance matrices.}
\end{quote}

Finally, we give fast algorithms for constructing approximate distance matrices. To establish some context, recall the classical Johnson-Lindenstrauss (JL) lemma which (roughly)
states that a random projection of a dataset $X \subset \R^d$ of
size $n$ onto a dimension of size $O(\log n)$ approximately
preserves all pairwise distances \cite{originalJL}. A common applications of this lemma is to \emph{instantiate} the $\ell_2$ distance matrix. A naive algorithm which computes the distance matrix after performing the JL projection requires approximately $O(n^2 \log n)$ time. Surprisingly, we show that the JL lemma is not tight with respect to creating an approximate $\ell_2$ distance matrix; we show that one can initialize the $\ell_2$ distance in an asymptotically better runtime.

\begin{theorem}[Informal; See Theorem \ref{thm:metric_compression} ]
 We can calculate a $n \times n$ matrix $B$ such that each $(i,j)$ entry $B_{ij}$ of $B$ satisfies $(1-\eps)\|x_i - x_j\|_2 \le B_{ij} \le   (1+\eps)\|x_i - x_j\|_2$
 in time $O( \eps^{-2}n^2 \, \log^2(\eps^{-1} \log n))$.
\end{theorem}

Our result can be viewed as the natural runtime bound which would follow if the JL lemma implied an embedding dimension bound of $O(\text{poly}(\log \log n))$. While this is impossible, as it would imply an exponential improvement over the JL bound which is tight \cite{larsen2017optimality}, we achieve our speedup by carefully reusing distance calculations via tools from metric compression \cite{indyk2017practical}. Our results also extend to the $\ell_1$ distance matrix; see Theorem \ref{thm:metric_compression} for details.

\paragraph{Notation.}
Our dataset will be the $n$ points $X = \{x_1 ,\ldots, x_n\} \subset \R^d$. For points in $X$, we denote $x_i(j)$ to be the $j$th coordinate of point $x_i$ for clarity. For all other vectors $v$, $v_i$ denotes the $i$th coordinate. We are interested in matrices of the form $A_{i,j} = f(x_i, x_j)$ for $f: \R^d \times \R^d \rightarrow \R$ which measures the similarity between any pair of points. $f$ might not necessarily be a distance function but we use the terminology ``distance function" for ease of notation. We will always explicitly state the function $f$ as needed. $w \approx 2.37$ denotes the matrix multiplication constant, i.e., the exponent of $n$ in the time required to compute the product of two $n \times n$ matrix \cite{alman2021refined}.

\subsection{Related Works}\label{sec:related_works}
\paragraph{Matrix-Vector Products Queries.}
% \piotr{We need to elaborate on multiplying by Gram matrices etc.} \piotr{We might want to ask Ryan Williams or Josh Alman for feedback.}
Our work can be understood as being part of a long line of classical works on the matrix free or implicit model as well as the active recent line of works on the matrix-vector query model. Many widely used linear algebraic algorithms such as the power method, the Lanczos algorithm \cite{lanczos1950iteration}, conjugate gradient descent \cite{shewchuk1994introduction}, and Wiedemann's coordinate recurrence algorithm \cite{wiedemann1986solving}, to name a few, all fall into this paradigm. Recent works such as \cite{musco2015randomized, braverman2020gradient, bakshi2022low} have succeeded in precisely nailing down the query complexity of these classical algorithms in addition to various other algorithmic tasks such as low-rank approximation \cite{ bakshi2022low}, trace estimation \cite{meyer2021hutch++}, and other linear-algebraic functions \cite{sun2021querying, rashtchian2020vector}. There is also a rich literature on query based algorithms in other contexts with the goal of minimizing the number of queries used. Examples include graph queries \cite{goldreich2017introduction}, distribution queries \cite{canonne2020survey}, and constraint based queries \cite{epstein2020property} in property testing, inner product queries in compressed sensing \cite{eldar2012compressed}, and quantum queries \cite{lee2021quantum, childs2021quantum}.

Most prior works on query based models assume black-box access to matrix-vector queries. While this is a natural model which allows for the design non-trivial algorithms and lower bounds, it is not always clear how such queries can be initialized. In contrast, the focus of our work is not on obtaining query complexity bounds, but rather complementing prior works by creating an efficient matrix-vector query for a natural class of matrices. 

\paragraph{Subquadratic Algorithms for Distance Matrices.} Most work on subquadratic algorithms for distance matrices have focused on the problem of computing a low-rank approximation. \cite{bakshi2018sublinear, indyk2019sample} both obtain an additive error low-rank approximation applicable for all distance matrices. These works only assume access to the entries of the distance matrix whereas we assume we also have access to the underlying dataset. \cite{bakshi2020robust} study the problem of computing the low-rank approximation of PSD matrices with also sample access to the entries of the matrix. Their results extend to low-rank approximation for the $\ell_1$ and $\ell_2^2$ distance matrices in addition to other more specialized metrics such as spherical metrics. Table \ref{tab:results2} lists the runtime comparisons between their results and ours.

Practically, the algorithm of \cite{indyk2019sample} is the easiest to implement and has outstanding empirical performance. We note that we can easily simulate their algorithm with no overall asymptotic runtime overhead using $O(\log n)$ vector queries. Indeed, their algorithm proceeds by sampling rows of the matrix according to their $\ell_2^2$ value and then post-processing these rows. The sampling probabilities only need to be accurate up to a factor of two. We can acquire these sampling probabilities by performing $O(\log n)$ matrix-vector queries which sketches the rows onto dimension $O(\log n)$ and preserves all row-norms up to a factor of two with high probability due to the Johnson-Lindenstrauss lemma \cite{originalJL}. This procedure only incurs an additional runtime of $O(T\log n)$ where $T$ is the time required to perform a matrix-vector query.

The paper \cite{indyk2004closest} shows that the exact $L_1$ distance matrix can be created in time $O(n^{(w+3)/2}) \approx n^{2.69}$ in the case of $d = n$, which is asymptotically faster than the naive bound of $O(n^2d) = O(n^3)$. In contrast, we focus on creating an (entry-wise) approximate distance matrices for all values of $d$.

We also compare to the paper of \cite{AlmanCS020}. In summary, their main upper bounds are approximation algorithms while we mainly focus on exact algorithms. Concretely, they study matrix vector products for matrices of the form $A_{i,j} = f(\|x_i - x_j \|_2^2)$
for some function $f: \R \rightarrow \R$. They present results on approximating the matrix vector product of $A$ where the approximation error is additive. They also consider a wide range of $f$, including polynomials and other kernels, but the input to is always the $\ell_2$ distance squared. In contrast, we also present exact algorithms, i.e., with no approximation errors. For example one of our main upper bounds is an exact algorithm when $A_{i,j} = \|x_i - x_j \|_1$ (see Table 1 for the full list). Since it is possible to approximately embed the $\ell_1$ distance into $\ell_2^2$, their methods could be used to derive approximate algorithms for $\ell_1$, but not the exact ones. Furthermore, we also study a wide variety of other distance functions such as $\ell_{\infty}$ and $\ell_p^p$ (and others listed in Table 1) which are not studied in Alman et al. In terms of technique, the main upper bound technique of Alman et al. is to expand
$f(\|x_i - x_j \|_2^2)$ and approximate the resulting quantity via a polynomial. This is related to our upper bound results for $\ell_p^p$ for even $p$ where we also use polynomials. However, our results are exact, while theirs are approximate. Our $\ell_1$ upper bound technique is orthogonal to the polynomial approximation techniques used in Alman et al. We also employ polynomial techniques to give upper bounds for the approximate $\ell_{\infty}$ distance function which is not studied in Alman et al. Lastly, Alman et al. also focus on the Laplacian matrix of the weighted graph represented by the distance matrix, such as spectral sparsification and Laplacian system solving. In contrast, we study different problems including low-rank approximations, eigenvalue estimation, and the task of initializing an approximate distance matrix. We do not consider the distance matrix as a graph or consider the associated Laplacian matrix. 
%We will include a detailed comparison to this reference in the updated version of the paper.

It is also easy to verify the ``folklore" fact that for a gram matrix $AA^T$, we can compute $AA^Tv$ in time $O(nd)$ if $A \in \R^{n \times d}$ by computing $A^Tv$ first and then $A(A^Tv)$. Our upper bound for the $\ell_2^2$ function can be reduced to this folklore fact by noting that $\|x-y\|_2^2 = \|x\|_2^2 + \|y\|_2^2 - 2\langle x,y\rangle$. Thus the $\ell_2^2$ matrix can be decomposed into two rank one components due to the terms $\|x\|_2^2$ and $\|y\|_2^2 $, and a gram matrix due to the term $\langle x,y\rangle$. This decomposition of the $\ell_2^2$ matrix is well-known (see Section $2$ in \cite{dokmanic2015euclidean}). Hence, a matrix-vector query for the $\ell_2^2$ matrix easily reduces to the gram matrix case. Nevertheless, we explicitly state the $\ell_2^2$ upper bound for completeness since we also consider all $\ell_p^p$ functions for any integer $p \ge 1$.

\paragraph{Polynomial Kernels.}
There have also been works on faster algorithms for approximating a kernel matrix $K$ defined as the $n \times n$ matrix with entries $K_{i,j} = k(x_i, x_j)$ for a kernel function $k$. Specifically for the polynomial kernel $k(x_i, x_j) = \langle x_i , x_j \rangle^p$, recent works such as \cite{avron2014subspace, ahle2020oblivious, woodruff2020near, song2021fast} have shown how to find a sketch $K'$ of $K$ which approximately satisfies $\|K'z\|_2 \approx \|Kz\|_2$ for all $z$. In contrast, we can exactly simulate the matrix-vector product $Kz$. Our runtime is $O(nd^p)$ which has a linear dependence on $n$ but an exponential dependence on $p$ while the aforementioned works have at least a quadratic dependence on $n$ but a polynomial dependence on $p$. Thus our results are mostly applicable to the setting where our dataset is large, i.e. $n \gg d$ and $p$ is a small constant. For example, $p = 2$ is a common choice in practice \cite{chang2010training}. Algorithms with polynomial dependence in $d$ and $p$ but quadratic dependence in $n$ are suited for smaller datasets which have very large $d$ and large $p$. Note that a large $p$ might arise if approximates a non-polynomial kernel using a polynomial kernel via a taylor expansion. We refer to the references within \cite{avron2014subspace, ahle2020oblivious, woodruff2020near, song2021fast} for additional related work. There is also work on kernel density estimation (KDE) data structures which upon query $y$, allow for estimation of the sum $\sum_{x \in X} k(x,y)$ in time sublinear in $|X|$ after some preprocessing on the dataset $X$. For widely used kernels such as the Gaussian and Laplacian kernels, KDE data structures were used in \cite{backurs2021} to create a matrix-vector query algorithm for kernel matrices in time subquadratic in $|X|$ for input vectors which are entry wise non-negative. We refer the reader to \cite{charikar2017hashing,backurs2018, siminelakis2019rehashing, backurs2019space,charikar2020} for prior works on KDE data structures.

\section{Faster Matrix-Vector Product Queries for \texorpdfstring{$\ell_1$}{L1}}
We derive faster matrix-vector queries for distance matrices for a wide array of distance metrics. First we consider the case of the $\ell_1$ metric such that $A_{i,j} = f(x_i, x_j)$ where $f(x,y) = \|x-y\|_1 = \sum_{i = 1}^d |x_i - y_i|$.

\begin{algorithm}[H]
\caption{\label{alg:preprocessing_1}Preprocessing}
\begin{algorithmic}[1]
\State \textbf{Input:} Dataset $X \subset \R^d$
\Procedure{Preprocessing}{}
\For{$i \in [d]$}
\State $T_i \gets $ sorted array of the $i$th coordinates of all $x \in X$.
\EndFor
\EndProcedure
\end{algorithmic}
\end{algorithm}

We first analyze the correctness of Algorithm \ref{alg:query_1}. 

\begin{theorem}\label{thm:p=1}
Let $A_{i,j} = \|x_i-x_j\|_1$. Algorithm \ref{alg:query_1} computes $Ay$ exactly.
\end{theorem}
\begin{proof}
Consider any coordinate $k \in [n]$. We show that $(Ay)_k$ is computed exactly. We have
\[(Ay)(k) = \sum_{j = 1}^n y_j \| x_k - x_j \|_1 = \sum_{j=1}^n \sum_{i = 1}^d y_j | x_k(i) - x_j(i)| = \sum_{i=1}^d \sum_{j=1}^n y_j |x_k(i)-x_j(i)|. \]
Let $\pi^i$ denote the order of $[n]$ induced by $T_i$. We have 
\[ \sum_{i=1}^d \sum_{j=1}^n y_j |x_k(i)-x_j(i)| = \sum_{i=1}^d \left( \sum_{j: \pi^i(k) \le \pi^i(j) } y_j(x_j(i) - x_k(i)) + \sum_{j: \pi^i(k) > \pi^i(j) } y_j(x_k(i) - x_j(i))  \right).\]
We now consider the inner sum. It rearranges to the following:
\begin{align*}
    &x_k(i)\left(\sum_{j: \pi^i(k) > \pi^i(j) } y_j- \sum_{j: \pi^i(k) \le \pi^i(j) } y_j   \right) +  \sum_{j: \pi^i(k) \le \pi^i(j) } y_j x_j(i) - \sum_{j: \pi^i(k) > \pi^i(j) } y_j x_j(i) \\
    &= x_k(i) \cdot ( S_3 - S_4 ) + S_2 - S_1,
\end{align*}
where $S_1, S_2, S_3,$ and $S_4$ are defined in lines $15-18$ of Algorithm \ref{alg:query_1} and the last equality follows from the definition of the arrays $B_i$ and $C_i$.
Summing over all $i \in [d]$ gives us the desired result.
\end{proof}

The following theorem readily follows.
\begin{theorem}\label{thm:l1_upper_bound}
Suppose we are given a dataset $\{x_1, \ldots, x_n\}$ which implicitly defines the distance matrix $A_{i,j} = \|x_i - x_j\|_1$. Given a query $y \in \R^d$, we can compute $Ay$ exactly in $O(nd)$ query time. We also require a one time $O(nd\log n)$ preprocessing time which can be reused for all queries.
\end{theorem}

\begin{algorithm}[H]
\caption{\label{alg:query_1}matrix-vector Query for $p=1$}
\begin{algorithmic}[1]
\State \textbf{Input:} Dataset $X \subset \R^d$
\State \textbf{Output:} $z = Ay$
\Procedure{Query}{$\{T_i\}_{i \in [d]}$, $y$}
\State $y_1, \cdots, y_n \gets$ coordinates of $y$.
\State Associate every $x_i \in X$ with the scalar $y_i$
\For{$i \in [d]$}
\State Compute two arrays $B_i, C_i$ as follows.
\State $B_i$ contains the partial sums of $y_jx_j(i)$ computed in the order induced by $T_i$
\State $C_i$ contains the partial sums of $y_j$ computed in the order induced by $T_i$
\EndFor
\State $z \gets 0^n$
\For{$k \in [n]$}
\For{$i \in [d]$}
\State $q \gets$ position of $x_k(i)$ in the order of $T_i$
\State $S_1 \gets B_i[q]$
\State $S_2 \gets B_i[n] - B_i[q]$
\State $S_3 \gets C_i[q]$
\State $S_4 \gets C_i[n] - C_i[q]$
\State $z(k) += x_k(i) \cdot ( S_3 - S_4 ) + S_2 - S_1$
\EndFor
\EndFor
\EndProcedure
\end{algorithmic}
\end{algorithm}

\section{Lower and Upper bounds for \texorpdfstring{$\ell_{\infty}$}{L-Infinity}}\label{sec:l_infinity_lowerbound}
In this section we give a proof of Theorem~\ref{thm:l_infinity_lower_bound}. Specifically, we give a reduction from a so-called {\em Orthogonal Vector Problem} (OVP)~\cite{williams2005new} to the problem of computing matrix-vector product $Az$, where $A_{i,j}=\|x_i-x_j\|_{\infty}$, for a given set of points $X=\{x_1, \ldots, x_n\}$. The orthogonal vector problem is defined as follows: given two sets of vectors $A=\{a^1, \ldots ,a^n\}$ and $B=\{b^1, \ldots ,b^n\}$, $A,B  \subset \{0,1\}^d$, $|A| = |B| = n$, determine whether there exist $x \in A$ and $y \in B$ such that the dot product $x \cdot y = \sum_{j=1}^d x_j y_j$ (taken over reals) is equal to 0. It is known that if OVP can be solved in strongly subquadratic time $O(n^{2-\alpha})$ for any constant $\alpha>0$ and $d=\omega(\log n)$, then SETH is false. Thus, an efficient reduction from OVP to the matrix-vector product problem yields Theorem \ref{thm:l_infinity_lower_bound}.

\begin{lemma}
If the matrix-vector product problem for $\ell_{\infty}$ distance matrices induced by $n$ vectors of dimension $d$ can be solved in time $T(n,d)$, then OVP (with the same parameters) can be solved in time $O(T(n,d))$.
\end{lemma}
\begin{proof}
Define two functions, $f,g:\{0,1\}^d \to [0,1]$, such that $f(0)=g(0)=1/2$, $f(1)=0$, $g(1)=1$. We extend both functions to vectors by applying $f$ and $g$ coordinate wise and to sets by letting $f(\{a^1, \ldots, a^n \}) = \{f(a^1), \ldots, f(a^n) \})$; the function $g$ is extended in the same way for $B$. Observe that, for any pair of non-zero vectors $a,b \in \{0,1\}^d$, we have $\|f(a)-g(b)\|_{\infty}=1$ if and only if $a\cdot b >0$, and $\|f(a)-g(b)\|_{\infty}=1/2$ otherwise.

Consider two sets of binary vectors $A$ and $B$. Without loss of generality we can assume that the vectors are non-zero, since otherwise the problem is trivial.  Define three distance matrices: matrix $M_A$ defined by the set $f(A)$, matrix $M_B$ defined by the set $g(B)$ and $M_{AB}$ defined by the set $f(A) \cup f(B)$. 
Furthermore, let $M$ be the ``cross-distance'' matrix, such that such that $M_{i,j}=\|f(a^i)-g(b^j)\|_{\infty}$. Observe that  the matrix $M_{AB}$ contains blocks $M_A$ and $M_B$ on its diagonal, and blocks $M$ and $M^T$ off-diagonal. Thus, $M_{AB} \cdot 1 = M_A\cdot 1 + M_B \cdot 1 + 2 M \cdot 1$, where $1$ denotes an all-ones vector of the appropriate dimension. Since $M \cdot 1 = (M_{AB} \cdot 1 -  M_A\cdot 1 - M_B \cdot 1)/2$, we can calculate $M \cdot 1$ in time $O(T(n,d))$. Since all  entries of $M$ are either $1$ or $1/2$, we have that $M \cdot 1 < n^2$ if and only if there is an entry $M_{i,j}=1/2$. However, this only occurs if $a^i \cdot b^j=0$.
\end{proof}

\subsection{Approximate \texorpdfstring{$\ell_{\infty}$}{L-Infinity} Matrix-Vector Queries}

In light of the lower bounds given above, we consider initializing \emph{approximate} matrix-vector queries for the $\ell_{\infty}$ function. Note that the lower bound holds for points in $\{0,1,2\}^d$ and thus it is natural to consider approximate upper bounds for the case of limited alphabet.  

\paragraph{Binary Case.}
We first consider the case that all points $x \in X$ are from $\{0,1\}^d$. We first claim the existence of a polynomial $T$ with the following properties. Indeed, the standard Chebyshev polynomials satisfy the following lemma, see e.g., see Chapter 2 in \cite{sachdeva2014faster}.

\begin{lemma}\label{lem:cheb}
There exists a polynomial $T: \R \rightarrow \R$ of degree $O(\sqrt{d} \log(1/\eps))$ such that $T(0) = 0$ and $|T(x)-1| \le \eps$ for all $x \in [1/d, 1]$.
\end{lemma}
Now note that $\|x-y\|_{\infty}$ can only take on two values, $0$ or $1$. Furthermore, $\|x-y\|_{\infty}  = 0$ if and only if $\|x-y\|_2^2 = 0$ and $\|x-y\|_{\infty}  = 1$ if and only if $\|x-y\|_2^2 \ge 1$. Therefore, $\|x-y\|_{\infty}  = 0$ if and only if $T(\|x-y\|_2^2/d) = 0$ and $\|x-y\|_{\infty}  = 1 $ if and only if $|T(\|x-y\|_2^2/d) - 1|\le \eps$.
Thus, we have that 
\[ |A_{i,j} - T(\|x_i-x_j\|_2^2/d)| = |\|x_i-x_j\|_{\infty} - T(\|x_i-x_j\|_2^2/d)| \le \eps \]
for all entries $A_{i,j}$ of $A$.
Note that $T(\|x-y\|_2^2/d)$ is a polynomial with $O((2d)^t)$ monomials in the variables $x(1), \ldots, x(d)$. Consider the matrix $B$ satisfying $B_{i,j} = T(\|x_i-x_j\|_2^2/d)$. Using the same ideas as our upper bound results for $f(x,y) = \langle x, y \rangle^p$, it is straightforward to calculate the matrix vector product $By$ (see Section \ref{sec:poly}). To summarize, for each $k \in [n]$, we write the $k$th coordinate of $By$ as a polynomial in the $d$ coordinates of $x_k$. This polynomial has $O((2d)^t)$ monomials and can be constructed in $O(n(2d)^t)$ time. Once constructed, we can evaluate the polynomial at $x_1 ,\ldots, x_n$ to obtain all the $n$ coordinates of $By$. Each evaluation requires $O((2d)^t)$ resulting in an overall time bound of $O(n(2d)^t)$.

\begin{theorem}
Let $A_{i,j} = \|x_i - x_j \|_{\infty}$. We can compute $By$ in time $O(n(2d)^{\sqrt{d} \log(1/\eps)})$ where $\|A-B\|_{\infty} \le \eps$.
\end{theorem}

\paragraph{Entries in $\{0, \ldots, M \}$.}
We now consider the case that all points $x \in X$ are from $\{0,\ldots, M\}^d$. Our argument will be a generalization of the previous section. At a high level, our goal is to detect which of the $M+1$ possible values in $\{0, \ldots, M\}$ is equal to the $\ell_{\infty}$ norm. To do so, we appeal to the prior section and design estimators which approximate the indicator function $``\|x-y\|_{\infty} \ge i"$. By summing up these indicators, we can approximate $\|x-y\|_{\infty}$.

Our estimators will again be designed using the Chebyshev polynomials. To motivate them, suppose that we want to detect if $\|x-y\|_{\infty} \ge i$ or if $\|x-y\|_{\infty} < i$. In the first case, some entry in $x-y$ will have absolute value value at least $i$ where as in the other case, all entries of $x-y$ will be bounded by $i-1$ in absolute value. Thus if we can boost this `signal', we can apply a polynomial which performs thresholding to distinguish the two cases. This motivates considering the functions of $\|x-y\|_k^k$ for a larger power $k$. In particular, in the case that $\|x-y\|_{\infty} \ge i$, we have $\|x-y\|_k^k \ge i^k$ and otherwise,  $\|x-y\|_k^k \le di^{k-1}$. By setting $k \approx \log(d)$, the first value is much larger than the latter, which we can detect using the `threshold' polynomials of the previous section.

We now formalize our intuition. It is known that appropriately scaled Chebyshev polynomials satisfy the following guarantees, see e.g., see Chapter 2 in \cite{sachdeva2014faster}.

\begin{lemma}\label{lem:cheb2}
There exists a polynomial $T: \R \rightarrow \R$ of degree $O(\sqrt{t} \log(t/\eps))$ such that $|T(x)| \le \eps/t$ for all $x \in [0, 1/(10t)]$ and $|T(x)-1| \le \eps/t^2$ for all $x \in [1/t, 1]$.
\end{lemma}
Given $x,y \in \R^d$, our estimator will first try to detect if $\|x-y\|_{\infty} \ge i$. Let $T_1$ be a polynomial from Lemma \ref{lem:cheb2} with $t = O(M^k)$ for $k = O(M \log(Md))$ and assuming $k$ is even. Let $T_2$ be a polynomial from Lemma \ref{lem:cheb2} with $t = O(\sqrt{d}\log(M/\eps))$. Our estimator will be 
\[T_2 \left( \frac{1}d \sum_{j=1}^d T_1 \left( \frac{(x(j)-y(j))^k}{i^k \cdot M^k} \right) \right). \]
If coordinate $j$ is such that $|x(j)-y(j)| \ge i$, then 
\[ \frac{(x(j)-y(j))^k}{i^k \cdot M^k} \ge \frac{1}{M^k} \]
and so $T_1$ will evaluate to a value very close to $1$. Otherwise, we know that 
\[ \frac{(x(j)-y(j))^k}{i^k \cdot M^k} \le \frac{(i-1)^k}{i^k M^k} = \frac{1}{M^k}\left(1 - 1/i \right)^k \ll \frac{1}{M^k} \cdot \frac{1}{\text{poly}(M, d)} \] by our choice of $k$, which means that $T_1$ will evaluate to a value close to $0$. Formally,
\[\frac{1}d \sum_{j=1}^d T_1 \left( \frac{(x(j)-y(j))^k}{i^k \cdot M^k} \right) \]
will be at least $1/d$ if there is a $j \in [d]$ with $|x(j)-y(j)| \ge i$ and otherwise, will be at most $1/(10d)$. By our choice of $T_2$, the overall estimate output at least $1-\eps$ in the first case and a value at most $\eps$ in the second case. 

The polynomial which is the concatenation of $T_2$ and $T_1$ has $O \left( \left( dk \cdot \text{deg}(T_1) \right)^{\text{deg}(T_2)} \right) =  (dM)^{O(M \sqrt{d} \log(Md))}$
monomials, if we consider the expression as a polynomial in the variables $x(1), \ldots, x(d)$. Our final estimator will be the sum across all $i \ge 1$. Following our upper bound techniques for matrix-vector products for polynomial, e.g. in Section \ref{sec:poly}, and as outlined in the prior section, we get the following overall query time:

\begin{theorem}Suppose we are given $X = \{x_1,\ldots,x_n\} \subseteq \{0,\ldots, M\}^d$ which implicitly defines the matrix $A_{i,j} = \|x_i - x_j\|_{\infty}$. For any query $y$, we can compute $By$ in time $n \cdot (dM)^{O(M \sqrt{d} \log (Md/\eps))}$ where $\|A-B\|_{\infty} \le \eps$.
\end{theorem}

\section{Empirical Evaluation}
We perform an empirical evaluation of our matrix-vector query for the $\ell_1$ distance function. We chose to implement our $\ell_1$ upper bound since it's a clean algorithm which possesses many of the same underlying algorithmic ideas as some of our other upper bound results. We envision that similar empirical results hold for most of our upper bounds in Table \ref{tab:results}. Furthermore, matrix-vector queries are the dominating subroutine in many key practical linear algebra algorithms such as the power method for eigenvalue estimation or iterative methods for linear regression: a fast matrix-vector query runtime automatically translates to faster algorithms for downstream applications.

\begin{table}[ht]
\centering
{\renewcommand{\arraystretch}{1.3}
\begin{tabular}{c|c|c|c|c}
Dataset                           & $(n,d)$                                 & Algo.                     & Preprocessing Time & Avg. Query Time                \\ \hline
\multirow{2}{*}{Gaussian Mixture} & \multirow{2}{*}{$(5 \cdot 10^4, 50$)}   & Naive                     & 453.7 s            & 43.3 s                         \\
                                  &                                         & Ours                      & 0.55 s             & 0.09 s                         \\ \hline
\multirow{2}{*}{MNIST}            & \multirow{2}{*}{$(5 \cdot 10^4, 784)$}  & Naive                     & 2672.5 s           & 38.6 s                         \\
                                  &                                         & \multicolumn{1}{l|}{Ours} & 5.5 s              & 1.9 s                          \\ \hline
\multirow{2}{*}{Glove}            & \multirow{2}{*}{$(1.2 \cdot 10^6, 50)$} & Naive                     & -                  & $\approx$ 2.6 days (estimated) \\
                                  &                                         & \multicolumn{1}{l|}{Ours} & 16.8 s             & 3.4 s                         
\end{tabular}
}
\caption{\label{table:experiments} Dataset description and empirical results. $(n,d)$ denotes the number of points and dimension of the dataset, respectively. Query times are averaged over $10$ trials with Gaussian vectors as queries.}
\end{table}

\paragraph{Experimental Design.}
We chose two real and one synthetic dataset for our experiments. We have two ``small" datasets and one ``large" dataset. The two small datasets have $5\cdot 10^4$ points whereas the large dataset has approximately $10^6$ points. The first dataset is points drawn from a mixture of three spherical Gaussians in $\R^{50}$. The second dataset is the standard MNIST dataset \cite{lecun1998mnist} and finally, our large dataset is Glove word embeddings\footnote{Can be accessed here: \url{http://github.com/erikbern/ann-benchmarks/}} in $\R^{50}$ \cite{pennington2014glove}.

The two small datasets are small enough that one can feasibly initialize the full $n \times n$ distance matrix in memory in reasonable time. A $5 \cdot 10^4 \times 5 \cdot 10^4$ matrix with each entry stored using $32$ bits requires $10$ gigabytes of space. This is simply impossible for the Glove dataset as approximately $5.8$ terabytes of space is required to initialize the distance matrix (in contrast, the dataset itself only requires $< 0.3$ gigabytes to store).

The naive algorithm for the small datasets is the following: we initialize the full distance matrix (which will count towards preprocessing), and then we use the full distance matrix to perform a matrix-vector query. Note that having the full matrix to perform a matrix-vector product only helps the naive algorithm since it can now take advantage of optimized linear algebra subroutines for matrix multiplication and does not need to explicitly calculate the matrix entries. 
Since we cannot initialize the full distance matrix for the large dataset, the naive algorithm in this case will compute the matrix-vector product in a standalone fashion by generating the entries of the distance matrix on the fly. We compare the naive algorithm to our Algorithms \ref{alg:preprocessing_1} and \ref{alg:query_1}. 

Our experiments are done in a 2021 M1 Macbook Pro with 32 gigabytes of RAM. We implement all algorithms in Python 3.9 using Numpy with Numba acceleration to speed up all algorithms whenever possible.
\paragraph{Results.}
Results are shown in Table \ref{table:experiments}. We show preprocessing and query time for both the naive and our algorithm in seconds. The query time is averaged over $10$ trials using Gaussian vectors as queries. For the Glove dataset, it was infeasible to calculate even a single matrix-vector product, even using fast Numba accelerated code. We thus estimated the full query time by calculating the time on a subset of $5 \cdot 10^4$ points of the Glove dataset and extrapolating to the full dataset by multiplying the query time by $(n/(5 \cdot 10^4))^2$ where $n$ is the total number of points. We see that in all cases, our algorithm outperforms the naive algorithm in both preprocessing time and query time and the gains become increasingly substantial as the dataset size increases, as predicted by our theoretical results.

\paragraph{Acknowledgements.}
This research was supported by the NSF TRIPODS program
(award DMS-2022448), Simons Investigator Award, MIT-IBM Watson collaboration, GIST-
MIT Research Collaboration grant, and NSF
Graduate Research Fellowship under Grant No. 1745302.

\bibliographystyle{alpha}
\bibliography{main}

\newcommand{\etalchar}[1]{$^{#1}$}
\begin{thebibliography}{MMMW21}

\bibitem[ACSS20]{AlmanCS020}
Josh Alman, Timothy Chu, Aaron Schild, and Zhao Song.
\newblock Algorithms and hardness for linear algebra on geometric graphs.
\newblock In Sandy Irani, editor, {\em 61st {IEEE} Annual Symposium on
  Foundations of Computer Science, {FOCS} 2020, Durham, NC, USA, November
  16-19, 2020}, pages 541--552. {IEEE}, 2020.

\bibitem[AG11]{ascher2011first}
Uri~M Ascher and Chen Greif.
\newblock {\em A first course on numerical methods}.
\newblock SIAM, 2011.

\bibitem[AKK{\etalchar{+}}20]{ahle2020oblivious}
Thomas~D Ahle, Michael Kapralov, Jakob~BT Knudsen, Rasmus Pagh, Ameya
  Velingker, David~P Woodruff, and Amir Zandieh.
\newblock Oblivious sketching of high-degree polynomial kernels.
\newblock In {\em Proceedings of the Fourteenth Annual ACM-SIAM Symposium on
  Discrete Algorithms}, pages 141--160. SIAM, 2020.

\bibitem[ANW14]{avron2014subspace}
Haim Avron, Huy Nguyen, and David Woodruff.
\newblock Subspace embeddings for the polynomial kernel.
\newblock {\em Advances in neural information processing systems}, 27, 2014.

\bibitem[AW21]{alman2021refined}
Josh Alman and Virginia~Vassilevska Williams.
\newblock A refined laser method and faster matrix multiplication.
\newblock In {\em Proceedings of the 2021 ACM-SIAM Symposium on Discrete
  Algorithms (SODA)}, pages 522--539. SIAM, 2021.

\bibitem[BCIS18]{backurs2018}
Arturs Backurs, Moses Charikar, Piotr Indyk, and Paris Siminelakis.
\newblock Efficient density evaluation for smooth kernels.
\newblock {\em 2018 IEEE 59th Annual Symposium on Foundations of Computer
  Science (FOCS)}, pages 615--626, 2018.

\bibitem[BCW20]{bakshi2020robust}
Ainesh Bakshi, Nadiia Chepurko, and David~P Woodruff.
\newblock Robust and sample optimal algorithms for psd low rank approximation.
\newblock In {\em 2020 IEEE 61st Annual Symposium on Foundations of Computer
  Science (FOCS)}, pages 506--516. IEEE, 2020.

\bibitem[BCW22]{bakshi2022low}
Ainesh Bakshi, Kenneth~L Clarkson, and David~P Woodruff.
\newblock Low-rank approximation with $1/\varepsilon^{1/3}$ matrix-vector
  products.
\newblock {\em arXiv preprint arXiv:2202.05120}, 2022.

\bibitem[BHSW20]{braverman2020gradient}
Mark Braverman, Elad Hazan, Max Simchowitz, and Blake Woodworth.
\newblock The gradient complexity of linear regression.
\newblock In {\em Conference on Learning Theory}, pages 627--647. PMLR, 2020.

\bibitem[BIMW21]{backurs2021}
Arturs Backurs, Piotr Indyk, Cameron Musco, and Tal Wagner.
\newblock Faster kernel matrix algebra via density estimation.
\newblock In {\em Proceedings of the 38th International Conference on Machine
  Learning}, pages 500--510, 2021.

\bibitem[BIW19]{backurs2019space}
Arturs Backurs, Piotr Indyk, and Tal Wagner.
\newblock Space and time efficient kernel density estimation in high
  dimensions.
\newblock In {\em Advances in Neural Information Processing Systems 32: Annual
  Conference on Neural Information Processing Systems, NeurIPS}, pages
  15773--15782, 2019.

\bibitem[BW18]{bakshi2018sublinear}
Ainesh Bakshi and David Woodruff.
\newblock Sublinear time low-rank approximation of distance matrices.
\newblock {\em Advances in Neural Information Processing Systems}, 31, 2018.

\bibitem[Can20]{canonne2020survey}
Cl{\'e}ment~L Canonne.
\newblock A survey on distribution testing: Your data is big. but is it blue?
\newblock {\em Theory of Computing}, pages 1--100, 2020.

\bibitem[CC08]{cox2008multidimensional}
Michael~AA Cox and Trevor~F Cox.
\newblock Multidimensional scaling.
\newblock In {\em Handbook of data visualization}, pages 315--347. Springer,
  2008.

\bibitem[CHC{\etalchar{+}}10]{chang2010training}
Yin-Wen Chang, Cho-Jui Hsieh, Kai-Wei Chang, Michael Ringgaard, and Chih-Jen
  Lin.
\newblock Training and testing low-degree polynomial data mappings via linear
  svm.
\newblock {\em Journal of Machine Learning Research}, 11(4), 2010.

\bibitem[CHL21]{childs2021quantum}
Andrew~M Childs, Shih-Han Hung, and Tongyang Li.
\newblock Quantum query complexity with matrix-vector products.
\newblock {\em arXiv preprint arXiv:2102.11349}, 2021.

\bibitem[CKNS20]{charikar2020}
Moses Charikar, Michael Kapralov, Navid Nouri, and Paris Siminelakis.
\newblock Kernel density estimation through density constrained near neighbor
  search.
\newblock {\em 2020 IEEE 61st Annual Symposium on Foundations of Computer
  Science (FOCS)}, pages 172--183, 2020.

\bibitem[CS17]{charikar2017hashing}
Moses Charikar and Paris Siminelakis.
\newblock Hashing-based-estimators for kernel density in high dimensions.
\newblock In Chris Umans, editor, {\em 58th {IEEE} Annual Symposium on
  Foundations of Computer Science, {FOCS}}, pages 1032--1043, 2017.

\bibitem[DPRV15]{dokmanic2015euclidean}
Ivan Dokmanic, Reza Parhizkar, Juri Ranieri, and Martin Vetterli.
\newblock Euclidean distance matrices: essential theory, algorithms, and
  applications.
\newblock {\em IEEE Signal Processing Magazine}, 32(6):12--30, 2015.

\bibitem[EK12]{eldar2012compressed}
Yonina~C Eldar and Gitta Kutyniok.
\newblock {\em Compressed sensing: theory and applications}.
\newblock Cambridge university press, 2012.

\bibitem[ES20]{epstein2020property}
Rogers Epstein and Sandeep Silwal.
\newblock Property testing of lp-type problems.
\newblock In {\em 47th International Colloquium on Automata, Languages, and
  Programming (ICALP 2020)}. Schloss Dagstuhl-Leibniz-Zentrum f{\"u}r
  Informatik, 2020.

\bibitem[Gol17]{goldreich2017introduction}
Oded Goldreich.
\newblock {\em Introduction to property testing}.
\newblock Cambridge University Press, 2017.

\bibitem[HS93]{holm1993protein}
Liisa Holm and Chris Sander.
\newblock Protein structure comparison by alignment of distance matrices.
\newblock {\em Journal of molecular biology}, 233(1):123--138, 1993.

\bibitem[ILLP04]{indyk2004closest}
Piotr Indyk, Moshe Lewenstein, Ohad Lipsky, and Ely Porat.
\newblock Closest pair problems in very high dimensions.
\newblock In {\em International Colloquium on Automata, Languages, and
  Programming}, pages 782--792. Springer, 2004.

\bibitem[IP01]{impagliazzo2001complexity}
Russell Impagliazzo and Ramamohan Paturi.
\newblock On the complexity of k-sat.
\newblock {\em Journal of Computer and System Sciences}, 62(2):367--375, 2001.

\bibitem[IPZ01]{impagliazzo2001problems}
Russell Impagliazzo, Ramamohan Paturi, and Francis Zane.
\newblock Which problems have strongly exponential complexity?
\newblock {\em Journal of Computer and System Sciences}, 63(4):512--530, 2001.

\bibitem[IRW17]{indyk2017practical}
Piotr Indyk, Ilya Razenshteyn, and Tal Wagner.
\newblock Practical data-dependent metric compression with provable guarantees.
\newblock {\em Advances in Neural Information Processing Systems}, 30, 2017.

\bibitem[IVWW19]{indyk2019sample}
Pitor Indyk, Ali Vakilian, Tal Wagner, and David~P Woodruff.
\newblock Sample-optimal low-rank approximation of distance matrices.
\newblock In {\em Conference on Learning Theory}, pages 1723--1751. PMLR, 2019.

\bibitem[JL84]{originalJL}
W.~Johnson and J.~Lindenstrauss.
\newblock Extensions of lipschitz maps into a hilbert space.
\newblock {\em Contemporary Mathematics}, 26:189--206, 01 1984.

\bibitem[Kru64]{kruskal1964multidimensional}
Joseph~B Kruskal.
\newblock Multidimensional scaling by optimizing goodness of fit to a nonmetric
  hypothesis.
\newblock {\em Psychometrika}, 29(1):1--27, 1964.

\bibitem[Kru78]{kruskal1978multidimensional}
Joseph~B Kruskal.
\newblock {\em Multidimensional scaling}.
\newblock Number~11. Sage, 1978.

\bibitem[Kuc09]{kuczma2009introduction}
Marek Kuczma.
\newblock {\em An introduction to the theory of functional equations and
  inequalities: Cauchy's equation and Jensen's inequality}.
\newblock Springer Science \& Business Media, 2009.

\bibitem[Lan50]{lanczos1950iteration}
Cornelius Lanczos.
\newblock An iteration method for the solution of the eigenvalue problem of
  linear differential and integral operators.
\newblock 1950.

\bibitem[LeC98]{lecun1998mnist}
Yann LeCun.
\newblock The mnist database of handwritten digits.
\newblock {\em http://yann. lecun. com/exdb/mnist/}, 1998.

\bibitem[LN17]{larsen2017optimality}
Kasper~Green Larsen and Jelani Nelson.
\newblock Optimality of the johnson-lindenstrauss lemma.
\newblock In {\em 2017 IEEE 58th Annual Symposium on Foundations of Computer
  Science (FOCS)}, pages 633--638. IEEE, 2017.

\bibitem[LSZ21]{lee2021quantum}
Troy Lee, Miklos Santha, and Shengyu Zhang.
\newblock Quantum algorithms for graph problems with cut queries.
\newblock In {\em Proceedings of the 2021 ACM-SIAM Symposium on Discrete
  Algorithms (SODA)}, pages 939--958. SIAM, 2021.

\bibitem[MM15]{musco2015randomized}
Cameron Musco and Christopher Musco.
\newblock Randomized block krylov methods for stronger and faster approximate
  singular value decomposition.
\newblock {\em Advances in neural information processing systems}, 28, 2015.

\bibitem[MMMW21]{meyer2021hutch++}
Raphael~A Meyer, Cameron Musco, Christopher Musco, and David~P Woodruff.
\newblock Hutch++: Optimal stochastic trace estimation.
\newblock In {\em Symposium on Simplicity in Algorithms (SOSA)}, pages
  142--155. SIAM, 2021.

\bibitem[PSM14]{pennington2014glove}
Jeffrey Pennington, Richard Socher, and Christopher~D Manning.
\newblock Glove: Global vectors for word representation.
\newblock In {\em Proceedings of the 2014 conference on empirical methods in
  natural language processing (EMNLP)}, pages 1532--1543, 2014.

\bibitem[RWZ20]{rashtchian2020vector}
Cyrus Rashtchian, David~P Woodruff, and Hanlin Zhu.
\newblock Vector-matrix-vector queries for solving linear algebra, statistics,
  and graph problems.
\newblock {\em arXiv preprint arXiv:2006.14015}, 2020.

\bibitem[S{\etalchar{+}}94]{shewchuk1994introduction}
Jonathan~Richard Shewchuk et~al.
\newblock An introduction to the conjugate gradient method without the
  agonizing pain, 1994.

\bibitem[SRB{\etalchar{+}}19]{siminelakis2019rehashing}
Paris Siminelakis, Kexin Rong, Peter Bailis, Moses Charikar, and
  Philip~Alexander Levis.
\newblock Rehashing kernel evaluation in high dimensions.
\newblock In {\em Proceedings of the 36th International Conference on Machine
  Learning, {ICML}}, pages 5789--5798, 2019.

\bibitem[SV{\etalchar{+}}14]{sachdeva2014faster}
Sushant Sachdeva, Nisheeth~K Vishnoi, et~al.
\newblock Faster algorithms via approximation theory.
\newblock {\em Foundations and Trends{\textregistered} in Theoretical Computer
  Science}, 9(2):125--210, 2014.

\bibitem[SWYZ21a]{song2021fast}
Zhao Song, David Woodruff, Zheng Yu, and Lichen Zhang.
\newblock Fast sketching of polynomial kernels of polynomial degree.
\newblock In {\em International Conference on Machine Learning}, pages
  9812--9823. PMLR, 2021.

\bibitem[SWYZ21b]{sun2021querying}
Xiaoming Sun, David~P Woodruff, Guang Yang, and Jialin Zhang.
\newblock Querying a matrix through matrix-vector products.
\newblock {\em ACM Transactions on Algorithms (TALG)}, 17(4):1--19, 2021.

\bibitem[SY07]{so2007theory}
Anthony Man-Cho So and Yinyu Ye.
\newblock Theory of semidefinite programming for sensor network localization.
\newblock {\em Mathematical Programming}, 109(2):367--384, 2007.

\bibitem[TSL00]{tenenbaum2000global}
Joshua~B Tenenbaum, Vin~de Silva, and John~C Langford.
\newblock A global geometric framework for nonlinear dimensionality reduction.
\newblock {\em science}, 290(5500):2319--2323, 2000.

\bibitem[Wie86]{wiedemann1986solving}
Douglas Wiedemann.
\newblock Solving sparse linear equations over finite fields.
\newblock {\em IEEE transactions on information theory}, 32(1):54--62, 1986.

\bibitem[Wil05]{williams2005new}
Ryan Williams.
\newblock A new algorithm for optimal 2-constraint satisfaction and its
  implications.
\newblock {\em Theoretical Computer Science}, 348(2-3):357--365, 2005.

\bibitem[Woo14]{woodruff2014sketching}
David~P Woodruff.
\newblock Sketching as a tool for numerical linear algebra.
\newblock {\em arXiv preprint arXiv:1411.4357}, 2014.

\bibitem[WS06]{weinberger2006unsupervised}
Kilian~Q Weinberger and Lawrence~K Saul.
\newblock Unsupervised learning of image manifolds by semidefinite programming.
\newblock {\em International journal of computer vision}, 70(1):77--90, 2006.

\bibitem[WZ20]{woodruff2020near}
David Woodruff and Amir Zandieh.
\newblock Near input sparsity time kernel embeddings via adaptive sampling.
\newblock In {\em International Conference on Machine Learning}, pages
  10324--10333. PMLR, 2020.

\end{thebibliography}

%%%%%%%%%%%%%%%%%%%%%%%%%%%%%%%%%%%%%%%%%%%%%%%%%%%%%%%%%%%%

%%%%%%%%%%%%%%%%%%%%%%%%%%%%%%%%%%%%%%%%%%%%%%%%%%%%%%%%%%%%

\appendix
\newpage

\section{Omitted Upper Bounds for Faster Matrix-Vector Queries}

We now consider the case of $\ell_p^p$ for $p=2$. Generalizing the results of $p=1$ and $p=2$ allows us to handle general $\ell_p^p$ functions.

\begin{algorithm}[H]
\caption{\label{alg:query_2}matrix-vector Query for $p=2$}
\begin{algorithmic}[1]
\State \textbf{Input:} Dataset $X \subset \R^d$
\State \textbf{Output:} $z = Ay$
\Procedure{Query}{y}
\State $v \gets \sum_{i=j}^n y_j x_j$
\State $S_1 \gets \sum_{i=j}^n y_j^2$
\State $S_2 \gets \sum_{i=j}^n y_j^2 \|x\|_2^2$
\State $z \gets 0^n$
\For{$k \in [n]$}
\State $z(k) \gets S_1 \|x_k\|_2^2 + S_2 - 2 \langle x_k , v \rangle$
\EndFor
\EndProcedure
\end{algorithmic}
\end{algorithm}

\begin{theorem}\label{thm:l2_squared_upper_bound}
We can compute $Ay$ in $O(nd)$ query time.
\end{theorem}
\begin{proof}
The proof follows from the following calculation of the $k$th coordinate of $Ay$:
\[(Ay)(k) = \sum_{j=1}^n y_j \|x_k - x_j\|_2^2 = \|x_k\|_2^2 \left( \sum_{j=1}^n y_j^2 \right) + \sum_{j=1}^n y_j^2 \|x_j\|_2^2 - 2 \left \langle x_k, \sum_{j=1}^n y_j x_j \right \rangle. \qedhere  \]
\end{proof}

We can extend our results to general $\ell_p^p$ functions as well as a wide array of commonly used functions to measure (dis)similarity between vectors. For example, suppose the points $x_i$ represent a probability distribution over the domain $[n] := \{1,\ldots, n\}$. A widely used ``distance" function over distributions is the KL-divergence defined as   
\[f(x_i,x_j) = \text{D}_{\text{KL}}(x_i \,  \| \,  x_j) = \sum_{k \in [d]} x_i(k) \log(x_i(k)) - x_i(k) \log(x_j(k)) = -H(x_i) - \sum_{k \in [d]} x_i(k) \log(x_j(k)), \]
where $H$ is the entropy function. Our techniques extend to the KL-divergence as well.

\begin{algorithm}[H]
\caption{\label{alg:query_KL}matrix-vector Query for KL Divergence}
\begin{algorithmic}[1]
\State \textbf{Input:} Dataset $X \subset \R^d$
\State \textbf{Output:} $z = Ay$
\Procedure{Query}{ $y$}
\State $S_{i} \gets \sum_{j=1}^n y_j \log(x_j(i))$ for all $i \in [d]$
\State $H_i \gets H(x_i)$ for all $i \in [n]$
\State $Y \gets \sum_{j=1}^n y_j$
\State $z \gets 0^n$
\For{$k \in [n]$}
\State $z(k) \gets - H_k \cdot Y - \sum_{i=1}^d x_k(i) S_i$
\EndFor
\EndProcedure
\end{algorithmic}
\end{algorithm}

\begin{theorem}\label{thm:kl}
 We can compute $Ay$ in $O(nd)$ query time.
\end{theorem}

\begin{proof}
Note that computed all of $S_i$ and $H_i$ takes $O(nd)$ time. Now
\begin{align*}
   (Ay)(k) &= \sum_{j=1}^n y_j \text{D}_{\text{KL}}(x_k \,  \| \,  x_j) \\
   &= \sum_{j=1}^n -y_jH(x_k) - \sum_{j = 1}^n y_j \sum_{k=1}^d x_i(k) \log(x_j(k))  \\
   &= - H(x_k)\left( \sum_{j=1}^n y_j \right) - \sum_{k=1}^d \sum_{j=1}^ny_j x_i(k) \log(x_j(k)) \\
   &= - H_k \cdot Y - \sum_{k=1}^d x_i(k) S_k,
\end{align*}
as desired.
\end{proof}

\subsection{General  \texorpdfstring{$p$}{p}}
We now consider the case of a general non-negative even integer $p$. 
\begin{algorithm}[H]
\caption{\label{alg:query_p_even}matrix-vector Query for even $p$}
\begin{algorithmic}[1]
\State \textbf{Input:} Dataset $X \subset \R^d$
\State \textbf{Output:} $z = Ay$
\Procedure{Query}{y}
\State Compute all the values $S_{i,t}:= \sum_{j=1}^n x_j(i)^{p-t}(-1)^{p-t}$ for all $i \in [d]$ and $t \in \{0, \ldots, p\}$
\State $z \gets 0^n$
\For{$k \in [n]$}
\State $z(k) \gets \sum_{i=1}^d \sum_{t=1}^p x_k(i)^t S_{i,t}$
\EndFor
\EndProcedure
\end{algorithmic}
\end{algorithm}

\begin{theorem}\label{thm:even_p}
We can compute $Ay$ in $O(ndp)$ query time.
\end{theorem}
\begin{proof}
Consider the following calculation of the $k$th coordinate of $Ay$:
\begin{align*}
    (Ay)(k) &= \sum_{j=1}^n y_j \|x_k - x_j\|_p^p \\
    &= \sum_{j=1}^n y_j \sum_{i=1}^d (x_k(i)-x_j(i))^p \\
    &= \sum_{j=1}^n \sum_{i=1}^d \sum_{t=1}^p x_k(i)^t x_j(i)^{p-t}(-1)^{p-t} \\
    &= \sum_{i=1}^d \sum_{t=1}^p x_k(i)^t \sum_{j=1}^n x_j(i)^{p-t}(-1)^{p-t} \\
    &= \sum_{i=1}^d \sum_{t=1}^p x_k(i)^t S_{i,t}.
\end{align*}
Note that computing $S_{i,t}$ for all $i$ and $t$ takes $O(ndp)$ time. Then returning the value of $(Ay)_k$ takes $O(dp)$ time resulting in the claimed runtime.
\end{proof}

The case of a general non-negative odd integer $p$ follows in a straightforward manner by combining the above techniques with those of the $p=1$ case of Theorem \ref{thm:l1_upper_bound} so we omit the proof.

\begin{theorem}\label{thm:odd_p}
For odd integer $p$, we can compute $Ay$ in $O(nd \log n)$ preprocessing time and $O(ndp)$ query time.
\end{theorem}

\subsection{Other Distance Functions}
In this section we initialize matrix-vector queries for a wide variety of ``distance" functions.

\paragraph{`Mixed' $\ell_\infty$.}
We consider the case of a `permutation invariant' version of the $\ell_{\infty}$ norm defined as follows:
\[f(x,y) = \max_{i \in [d], j \in [d]} |x_i - y_j|. \]
$f$ is not a norm but we will refer to it as `mixed' $\ell_\infty$.

\begin{algorithm}[H]
\caption{\label{alg:preprocessing_mixed_infinity}Preprocessing}
\begin{algorithmic}[1]
\State \textbf{Input:} Dataset $X \subset \R^d$
\Procedure{Preprocessing}{}
\For{$j \in [n]$}
\State $\min_j, \max_j \gets $ minimum and maximum values of the entries of $x_j$, respectively.
\EndFor
\EndProcedure
\end{algorithmic}
\end{algorithm}

\begin{algorithm}[H]
\caption{\label{alg:query_mixed}matrix-vector Query for mixed $\ell_\infty$}
\begin{algorithmic}[1]
\State \textbf{Input:} Dataset $X \subset \R^d$
\State \textbf{Output:} $z = Ay$
\Procedure{Query}{$\{\min_j, \max_j\}_{j =1}^n, y$}
\State $z \gets 0^n$
\For{$k \in [n]$}
\State $z(k) \gets \sum_{j=1}^n y_j \cdot \max\left( |\min_k - \min_j|, |\min_k - \max_j|, |\max_k - \min_j|, |\max_k - \max_j|\right) $
\EndFor
\EndProcedure
\end{algorithmic}
\end{algorithm}

\begin{theorem}\label{thm:mixed}
We can compute $Ay$ in $O(nd)$ preprocessing time and $O(n^2)$ query time.
\end{theorem}
\begin{proof}
The preprocessing time holds because we calculate the maximum and minimum of a list of $d$ numbers a total of $n$ times. For the query time, note that each $z(k)$ takes $O(n)$ time to compute since we do a $O(1)$ operation is each index of the sum in Line 6 of Algorithm \ref{alg:query_mixed}.

To prove correctness, note that for any two vectors $x, y \in \R^d$, the maximum value of  $|x_i - y_j|$ is attained if $x_i$ and $y_j$ are among the minimum / maximum values of the coordinates of $x$ and $y$ respectively. To see this, fix a value of $x_i$. We can always increase $|x_i - y_j|$ by setting $y_j$ to be the maximum or minimum over all $j$.
\end{proof}

\paragraph{Mahalanobis Distance Squared.}
We consider the function 
\[f(x,y) = x^TMy \]
for some $d \times d$ matrix $M$. This is the squared version of the well-known Mahalanobis distance.

\begin{algorithm}[H]
\caption{\label{alg:preprocessing_mahalanobis}Preprocessing}
\begin{algorithmic}[1]
\State \textbf{Input:} Dataset $X \subset \R^d$
\Procedure{Preprocessing}{}
\State $S \gets d \times n$ matrix where the $j$th column is $Mx_j$ for all $j \in [n]$.
\EndProcedure
\end{algorithmic}
\end{algorithm}

\begin{algorithm}[H]
\caption{\label{alg:query_mahalanobis}matrix-vector Query for Mahalanobis distance squared}
\begin{algorithmic}[1]
\State \textbf{Input:} Dataset $X \subset \R^d$
\State \textbf{Output:} $z = Ay$
\Procedure{Query}{ $S, y$}
\State $v \gets Sy$
\State $z \gets 0^n$
\For{$k \in [n]$}
\State $z(k) \gets \langle x_k, v \rangle$
\EndFor
\EndProcedure
\end{algorithmic}
\end{algorithm}

\begin{theorem}\label{thm:maha}
We can compute $Ay$ in $O(nd^2)$ preprocessing time and $O(nd)$ query time.
\end{theorem}
\begin{proof}
Note that the $k$th coordinate of $Ay$ is given by
\[(Ay)(k) = \sum_{j=1}^n y_j x_k^T M x_j = \left \langle  x_k, \sum_{j=1}^n y_j M x_k \right \rangle = \langle x_k, Sy \rangle \]
which proves correctness. It takes $O(nd^2)$ time to compute $S$, $O(nd)$ time to compute $Sy$, and then $O(d)$ time to compute the $k$th coordinate of $Ay$ for all $k \in [n]$.
\end{proof}

\paragraph{Polynomial Kernels.}\label{sec:poly}
We now consider polynomial kernels of the form 
\[f(x,y) = \langle x, y \rangle^p. \]

\begin{theorem}\label{thm:poly_kernel}
We can compute $Ay$ in $O(nd^p)$ query time.
\end{theorem}
\begin{proof}[Proof Sketch]
Consider the following expression
\[ g(z) = \sum_{j=1}^n y_j \langle z, x_j \rangle^p \]
as a polynomial $g: \R^d \rightarrow \R$ in the $d$ coordinates of $z$. By rearranging, the above sum can be written as a sum over $O(d^p)$ terms, corresponding to each monomial $z_1^{a_1} \ldots z_d^{a_d}$ where $a_1 + \ldots + a_d = p$. The coefficient of each term takes $O(nd)$ time to compute given $x_i$ and $y$. Once computed, we can evaluate the polynomial at $z = x_j$ for all $j$ which form the coordinates of $Ay$. Again, this can be viewed as ``linearizing" the kernel given by $\langle x,y\rangle^p$.
\end{proof}

We note that a proof similar to that of Theorem \ref{thm:poly_kernel} was given in Section 5.3 of \cite{AlmanCS020} by expanding the relevant quantity as a polynomial; see Section \ref{sec:related_works} for detailed comparison between \cite{AlmanCS020} and our work.

\subsection{Distances for Distributions}
We now consider the case that each $x_i$ specifies a discrete distribution over a domain of $d$ elements. Matrices $A$ where $A_{i,j}$ is a function computing distances between distributions $x_i$ and $x_j$ have recently been studied in machine learning. 

We consider how to construct matrix-vector queries for such matrices for a range of widely used distance measures on distributions. First note that a result on the TV distance follows immediately from our $\ell_1$ result.

\begin{theorem}\label{thm:tv}
Suppose $A_{i,j} = \textup{TV}(x_i, x_j)$. We can compute $Ay$ in $O(nd \log n)$ preprocessing time and $O(nd)$ query time.
\end{theorem}

We now consider some other distance functions on distributions. 

\paragraph{Divergences.}

Through a similar calculation as the KL divergence case, we can also achieve $O(nd)$ query times if $f$ is the Jensen-Shannon divergence, defined as
\[ f(x,y) = \frac{\text{D}_{\text{KL}}(x \,  \| \,  y) + \text{D}_{\text{KL}}(y \,  \| \,  x)}2,\]
as well as the cross entropy function.

\begin{theorem}\label{thm:all_kl}
 Let $f$ be the Jensen-Shannon divergence or cross entropy function. Then $Ay$ can be computed in $O(nd)$ time.
\end{theorem}

Through a similar calculation as done in Section \ref{sec:poly} (for the case of $p=1$), we can also perform matrix-vector multiplication queries in the case that 
\[f(x,y) = \sum_{i = 1}^d \sqrt{x(i) y(i)}. \]
This is the squared Hellinger distance.

\begin{theorem}\label{thm:hellinger}
 Let $f$ be the squared Hellinger distance. Then $Ay$ can be computed in $O(nd)$ time.
\end{theorem}

\subsection{Approximate Matrix-Vector Query for \texorpdfstring{$\ell_2$}{L-2}}
While our techniques do not extend to the $\ell_2$ case for \emph{exact} matrix-vector queries, we can nonetheless instantiate \emph{approximate} matrix-vector queries for the $\ell_2$ function.
We first recall the following well known embedding result.

\begin{theorem}\label{thm:l2embed}
Let $\eps \in (0,1)$ and define $T: \R^d \rightarrow \R^k$ by 
\[T(x)_i = \frac{1}{\beta k} \sum_{j=1}^d Z_{ij}x_j, \quad i= 1, \ldots, k \]
where $\beta = \sqrt{2/\pi}$. Then for every vector $x \in \R^d$, we have
\[\Pr[(1-\eps)\|x\|_2 \le \|T(x)\|_1 \le (1+\eps) \|x\|_2] \ge 1-e^{c \eps^2 k}, \]
where $c > 0$ is a constant.
\end{theorem}

We can instantiate approximate matrix-vector queries for $f(x,y) = \|x-y\|_2$ via the following algorithm.

\begin{algorithm}[H]
\caption{\label{alg:preprocessing_l2}Preprocessing}
\begin{algorithmic}[1]
\State \textbf{Input:} Dataset $X \subset \R^d$
\Procedure{Preprocessing}{$T$}
\State $X' \gets TX$ where $T$ is the linear map from Theorem \ref{thm:l2embed}
\State Run Algorithm \ref{alg:preprocessing_1} on $X'$
\EndProcedure
\end{algorithmic}
\end{algorithm}

For queries, we just run Algorithm \ref{alg:query_1} on $X'$. We have the following guarantee:

\begin{theorem}\label{thm:l2_approx}
Let $A_{i,j} = \|x_i - x_j\|_2$. There exists a matrix $B$ such that we can compute $By$ in $O(nd^2 + nd \log n)$ preprocessing time and $O(n \log(n)/\eps^2)$ query time where
\[\|A-B\|_F \le \eps \|A\|_F \]
with probability $1-1/\textup{poly}(n)$.
\end{theorem}
\begin{proof}The preprocessing and query time follow from the time required to apply the transformation $T$ from Theorem \ref{thm:l2embed} to our set of points $X$ as well as the time needed for the $\ell_1$ matrix-vector query result of Theorem \ref{thm:p=1}. The Frobenius norm guarantee follows from the fact that every entry of $A$ will be approximated with relative error in $B$ using Theorem \ref{thm:l2embed}.
\end{proof}

\subsection{Matrix-Vector Query Lower Bounds}
Table \ref{tab:results} shows that we can initialize matrix-vector queries for a variety of distance functions in $O(nd)$ time. It is straightforward to see that this bound is optimal for a large class of distance matrices.

\begin{theorem}\label{thm:lb_for_ub}
Consider the case that $A_{i,j} = f(x_i, x_j)$ satisfying $f(x,x) = 0$ for all $x$. Further assume that for all $x$, there exists an input $y$ such that $f(x,y) =1$. An algorithm which outputs an entry-wise approximation of $Az$ to any constant factor for input $z$ requires $\Omega(nd)$ time in the worst case.
\end{theorem}
\begin{proof}
We consider two cases for input points of $A$. In the first case, all points in our dataset $X$ are identical. In the second case, a randomly chosen point is distance $1$ away from the $n-1$ identical points. Computing the product of $A$ times the all ones vector allows us to distinguish the two cases as $A1$ has entries summing to $0$ in the first case whereas $A1$ has entries summing to $n-1$ in the second case. Thus to approximate $A1$ entry-wise to any constant factor, we must distinguish the two cases. If we read $o(n)$ points, then with good probability we will see no duplicates. Thus, we must read $\Omega(n)$ points, require $\Omega(nd)$ time.
\end{proof}

\section{When Do Our Upper Bound Techniques Work?}\label{sec:meta}
By this point, we have seen many examples of matrix-vector queries which can be initialized as well as a lower bound for a natural distance function which prohibits any subquadratic time algorithm. Naturally, we are thus interested in the limits of our upper bound techniques for instantiating fast matrix-vector product. An understanding of such limits sheds light on families of structured matrices which may admit fast matrix-vector queries in general. In this section we fully characterize the capabilities of our upper bound methods and show that essentially our techniques can only work ``linear" functions (in a possibly different basis).

First we set some notation. Let $A$ be a $n \times n$ matrix we wish to compute where each $(i,j)$ entry is given by $f(x_i, x_j)$. Given a query vector $z \in \R^n$, the $k$th coordinate of $Az$ is given by
\[ (Az)(k) = \sum_{i=1}^n z_i f(x_k,x_i). \]
An example choice of $f$ is given by $f(x, y) = \sum_{j=1}^d x(j) \log(y(j))$ (assuming all the coordinates of $x$ and $y$ are entry wise non-negative. Note this is related to the cross entropy function in Table \ref{tab:results}).

We first highlight the major steps which are common to all of our upper bound algorithms using $f$ as an example. Our upper bound technique proceeds as follows:
\begin{itemize}
    \item Break up $f(x, y)$ into a sum over $d$ terms: $\sum_{j=1}^d x(j) \log(y(j))$.
    \item Switch the order of summation: 
    \[(Az)(k) = \sum_{i=1}^n z_i f(x_k,x_i) = \sum_{j=1}^d \sum_{i=1}^n z_ix_k(j)\log(x_i(j)).\]
    \item Evaluate each of the inner $d$ summations with $1$ evaluation each (after some preprocessing). In other words, for a fixed $j$, each of the sums $\sum_{i=1}^n z_ix_k(j)\log(x_i(j))$ can be computed as one evaluation, namely $x_k(j) \cdot \Big( \sum_{i=1}^n z_i\log(x_i(j)) \Big)$ and in preprocessing, we can compute $ \sum_{i=1}^n z_i\log(x_i(j))$ as it does not depend on the coordinate $k$.
\end{itemize}

The key steps of the above outline, namely switching the order of summation and precomputation of repeated terms, can be encapsulated in the following framework.

\begin{theorem}\label{thm:meta_upperbound}
Suppose there exist mappings $T_1, T_2: \R^d \rightarrow \R^{d'}$ (possibly non-linear) and a continuous $g: \R \times \R \rightarrow \R$ such that for every $k$,
\begin{align*}
    (Az)(k) &= \sum_{i=1}^n z_i f(x_k,x_i) \\
    &= \sum_{i=1}^n z_i \sum_{j=1}^{d'} g\left(T_1(x_k)(j), T_2(x_i)(j)\right) \quad \textup{(breaking $f$ into sum over $d'$ terms)} \\
    &= \sum_{j=1}^{d'} \sum_{i=1}^n z_i \cdot g\left(T_1(x_k)(j), T_2(x_i)(j)\right) \quad \textup{(switching order of summation).}
\end{align*}
Further suppose that each of the terms $\sum_{i=1}^n z_i\cdot g\left(T_1(x_k)(j), T_2(x_i)(j)\right)$ can be evaluated as
\[ \sum_{i=1}^n z_i \cdot g\left(T_1(x_k)(j), T_2(x_i)(j)\right) = g\left(T_1(x_k)(j), \sum_{i=1}^n z_i T_2(x_i)(j)\right) \]
for any choice of the vector $z$. Then $g(a,b)$ must be a linear function in $b$.
\end{theorem}
Theorem \ref{thm:meta_upperbound} is stated in quite general terms. We are stipulating the following statements: the functions $T_1, T_2$ represent possibly non-linear transformations to $\R^{d'}$ on $x,y$ respectively such that $f(x,y)$ can be decomposed as a sum over $d'$ function evaluations. Each function evaluation takes in the same coordinate, say the $j$th coordinate, of both $T_1(x)$ and $T_2(y)$ and computes the function $g(T_1(x)(j), T_2(y)(j))$. Finally the resulting sum $\sum_{i=1}^n z_i \cdot  g\left(T_1(x_k)(j), T_2(x_i)(j)\right)$ can be computed as  $g\Big (T_1(x_k)(j), \sum_{i=1}^n h(z_i)T_2(x_i)(j)\Big)$.

If these conditions hold (which is precisely the case in the proof of \emph{all} our upper bound results), then it \emph{must} be the case that $g$ has a very special form, in particular, $g$ must be a linear function in its second variable. To make the setting more concrete, we map the terminology of Theorem \ref{thm:meta_upperbound} into some examples from our upper bound results. 

First consider the case that $f(x,y) = \langle x, y \rangle$. In this case, both $T_1$ and $T_2$ are the identity maps and $g(a,b) = ab$. It is indeed the case that $g(a,b)$ is linear in $b$.  Now consider a slightly more complicated choice $f(x,y) = \sum_{j=1}^d x(j) \log(y(j))$. Here, we first have the mappings $T_1$ = identity but $T_2$ is a coordinate wise map such that $T_2(y) = [\log(y_1), \ldots, \log(y_n)]$. The function $g$ again satisfies $g(a,b) = ab$. Finally we consider the example $f(x,y) = \|x-y\|_2^2$ which sets $d' \gg d$. In particular, the mapppings $T_1, T_2$ expand $x,y$ into a $O(d^2)$-dimensional vector, whose coordinates represent all possible combinations products of two coordinates of $x$ and $y$ respectively. (The reader may realize that this particular case is an example of ``linearizing'' the kernel given by $f$). Again $g$ is the same function as before.

The proof of Theorem \ref{thm:meta_upperbound} relies on the following classical result on the solutions of Cauchy's functional equation.

\begin{theorem}\label{thm:cauchy}
Let $t: \R \rightarrow \R$ be a continuous function which satisfies $t(x+y) = t(x)+t(y)$ for all inputs $x,y$ in its domain. Then $t$ must be a linear function.
\end{theorem}

For us the hypothesis that $t$ is continuous suffices but it is know that the above result follows from much weaker hypothesis. We refer to the reader to \cite{kuczma2009introduction} for reference related to Cauchy's functional equation.

\begin{proof}[Proof of Theorem \ref{thm:meta_upperbound}]
Our goal is to show that if 
\[ \sum_{i=1}^n z_i \cdot g\left(T_1(x_k)(j), T_2(x_i)(j)\right) = g\left(T_1(x_k)(j), \sum_{i=1}^n z_i T_2(x_i)(j)\right) \]
for all $z$ and choices of input points $x_i$ then $g$ must be linear in the second variable.  First set $z_j = 0$ for all $j \ge 2$ and $z_1 = z_2 = 1$. For ease of notation, denote $q := T_1(x_k)(j)$. As we vary the coordinate of the points $x_1$ and $x_2$, the values $T_2(x_1)(j)$ and $T_2(x_2)(j)$ also vary over the range of $T_2$. Thus, 
\[g(q, a) + g(q,b) = g(q, a+b) \]
for all possible inputs $a,b$. However, this is exactly the hypothesis of Theorem \ref{thm:cauchy} so it follows that $g$ must be a linear function in its second coordinate, as desired.
\end{proof}

While the proof of Theorem \ref{thm:meta_upperbound} is straightforward, it precisely captures the scenarios where our upper bound techniques apply. In short, it implies that $f$ must have a linear structure, under a suitable change of basis, for our techniques to hold. If its not the case, then our techniques do not apply and new ideas are needed. Nevertheless, as displayed by the versatility of examples in Table \ref{tab:results}, such a structure is quite common in many applications where matrices of distance or similarity functions arise. 

The observant reader might wonder how our result for the $\ell_1$ function fits into the above framework as it is not obviously linear. However, we note that the function $h_j(x) = \sum_{i=1}^n |x(j)-x_i(j)|$ (which appears in the theorem statement of Theorem \ref{thm:meta_upperbound} as the sum $\sum_{i=1}^n z_i \cdot g\Big(T_1(x_k)(j), T_2(x_i)(j)\Big)$ is actually a \emph{piece-wise} linear function in $x(j)$. The sorting preprocessing we performed for Theorem \ref{thm:p=1} can be thought of as creating a data structure which allows us to efficiently index into the correct linear piece.

\section{Applications of Matrix-Vector Products}\label{sec:applications} 

\subsection{Preliminary Tools}
We highlight specific prior results which we use in conjunction with our matrix-vector query upper bounds to obtain improved algorithmic results. First we recall a result of \cite{bakshi2022low} which gives a nearly optimal low-rank approximation result in terms of the number of matrix-vector queries required.

\begin{theorem}[Theorem 5.1 in \cite{bakshi2022low}]\label{thm:bakshi22}
Given matrix-vector query access to a matrix $A \in \R^{n \times n}$, accuracy parameter $\eps \in (0,1), k \in [n]$ and any $p \ge 1$, there exists an algorithm which uses $\tilde{O}(k/\eps^{1/3})$ matrix-vector queries and outputs a $n \times k$ matrix $Z$ with orthonormal columns such that with probability at least $9/10$,
\[ \|A(I-ZZ^T)\|_p \le (1+\eps) \min_{U : U^TU = I_k} \|A(I-UU^T) \|_p \]
where $\|M\|_p = (\sum_{i=1}^n \sigma_i(M)^p)^{1/p}$ is the $p$-th Schatten norm where $\sigma_1 ,\ldots, \sigma(M)$ are the singular values of $M$. The runtime of the algorithm is $\tilde{O}(Tk/\eps^{1/3} + nk^{w-1}/\eps^{(w-1)/3})$ where $T$ is the time for computing a matrix-vector query.
\end{theorem}

The second result is of \cite{musco2015randomized} which give an optimized analysis of a variant of power method for computing the top $k$ singular values.

\begin{theorem}[Theorem $1$ and $7$ in \cite{musco2015randomized}]\label{thm:musco}
Given matrix-vector query access to a matrix $A \in \R^{n \times n}$, accuracy parameter $\eps \in (0,1), k \in [n]$, there exists an algorithm which uses $\tilde{O}(k/\eps^{1/2})$ matrix-vector queries and outputs a $1\pm \eps$ approximation to the top $k$ singular values of $A$. The runtime of the algorithm is 
$\tilde{O}(Tk/\eps^{1/2} + nk^2/\eps + k^3/\eps^{3/2})$.
\end{theorem}

Lastly, we recall the gaurantees of the classical conjugate-gradient descent method.

\begin{theorem}\label{thm:linear_system}
Let $A$ be a symmetric PSD matrix and consider the linear system $Ax = b$ and let $x^* = \text{argmin}_x \|Ax-b\|_2$. Let $\kappa$ denote the condition number of $A$. Given a starting vector $x_0$, the conjugate gradient descent algorithm uses $O(\sqrt{\kappa}\log(1/\eps))$ matrix-vector queries and returns $x$ such that
\[ \|x-x^*\|_A \le \eps \|x_0 - x^*\|_A \]
where $\|x\|_A = (x^TAx)^{1/2}$ denotes the $A$-norm.
\end{theorem}

Note that matrices in our setting are also PSD, for example if $A_{i,j} = \langle x_i, x_j \rangle$. For non PSD matrices $A$, one can also use the conjugate gradient descent method on the matrix $A^TA$ which squares the condition number. Therefore, there are more complicated algorithms which work directly on the matrix-vector queries of $A$ for non PSD matrices, for example see references in Chapters 6 and 7 of of \cite{ascher2011first}. We omit their discussion for clarity and just note that in practice, iterative methods which directly use matrix-vector queries are preferred for linear system solving. 

\subsection{Applications}

We now derive specific applications using prior results from ``matrix-free" methods. First we cover low-rank approximation.

For the $\ell_1$ and $\ell_2^2$ distance matrices, we improve upon prior works for computing a relative error low-rank approximation. While we can obtain such an approximation for a wide variety of Schatten norms, we state the bound in terms of the Frobenius norm since it has been studied in prior works.

\begin{theorem}\label{thm:opt_lowrank_p}
 Let $p \ge 1$ and consider the case that $A_{i,j} = \|x_i - x_j\|_p^p$ for all $i,j$. We can compute a matrix $B$ such that 
 \[\|A-B\|_F \le (1+\eps)\|A-A_k\|_F \]
 where $A_k$ denotes the optimal rank-$k$ approximation to $A$ in Frobenius norm. The runtime is $\tilde{O}(ndpk/\eps^{1/3} + nk^{w-1}/\eps^{(w-1)/3})$.
\end{theorem}
\begin{proof}
The theorem follows from combining the matrix-vector query runtime of Table \ref{tab:results} and Theorem \ref{thm:bakshi22}.
\end{proof}

Note that the best prior result for the special case of $\ell_1$ and $\ell_2^2$ from \cite{bakshi2020robust} where they obtained a runtime bound of $O(ndk/\eps + nk^{w-1}/\eps^{w-1})$. Thus our bound improves upon this by a multiplicative factor of $\text{poly}(1/\eps)$. We point out that the bound of $O(ndk/\eps + nk^{w-1}/\eps^{w-1})$ is actually \emph{optimal} for the class of algorithms which sample the entries of $A$. Thus our results show that if we know the set of points beforehand, which is a natural assumption, one can overcome such lower bounds.

For the case of $A_{i,j} = \|x_i - x_j\|_2$, we cannot hope to achieve a relative error approximation for low-rank approximation since we only have fast matrix-vector queries to the matrix $B$ where $B_{i,j} = (1 \pm \eps) \|x_i - x_j\|_2$ via Theorem \ref{thm:l2_approx}. Nevertheless, we can still obtain an additive error low-rank approximation guarantee which outperforms prior works. First we show that our approximation matrix-vector queries are sufficient to obtain such a guarantee.

\begin{lemma}\label{lem:approx_low_rank}
Let $A,B$ satisfy $\|A-B\|_F \le \eps \|A\|_F$ and suppose $A'$ and $B'$ are the best rank-$r$ approximation of $A$ and $B$ respectively in the Frobenius norm. Then 
\[\|A-B'\|_F \le \|A-A'\|_F + 2 \eps \|A\|_F. \]
\end{lemma}
\begin{proof}
We have
\begin{align*}
    \|A-B'\|_F &\le \|A-B\|_F + \|B-B'\|_F \\
    &\le \eps \|A\|_F + \|B-A'\|_F \\
    &\le \eps \|A\|_F + \|B-A\|_F + \|A-A'\|_F \\
    &\le \|A-A'\|_F + 2 \eps \|A\|_F. \qedhere
\end{align*}
\end{proof}

\begin{theorem}\label{thm:approx_low_rank}
Let $A_{i,j} = \|x_i - x_j\|_2$ for all $i,j$. We can compute a matrix $B$ such that 
 \[\|A-B\|_F \le \|A-A_k\|_F + \eps \|A\|_F \]
 with probability $1-1/\textup{poly}(n)$ where $A_k$ denotes the optimal rank-$k$ approximation to $A$ in Frobenius norm. The runtime is $\tilde{O}(ndk/\eps^{1/3} + nk^{w-1}/\eps^{(w-1)/3})$.
\end{theorem}
\begin{proof}
The runtime follows from applying Theorem \ref{thm:bakshi22} on the matrix created after applying Theorems \ref{thm:l2embed} and \ref{thm:l2_approx}. The approximation guarantee follows from Lemma \ref{lem:approx_low_rank}.
\end{proof}

The best prior work for additive error low-rank approximation for this case is due to \cite{indyk2019sample} which obtained such a guarantee with runtime $\tilde{O}(nd \cdot \text{poly}(k, 1/\eps))$ for a large unspecified polynomial in $k$ and $1/\eps$. Lastly we note that our relative error low-rank approximation guarantee holds for any $f$ in Table \ref{tab:results}, as summarized in Table \ref{tab:results2}.

\begin{theorem}\label{thm:low_rank_meta}
 Suppose we have exact matrix-vector query access to a matrix $A$ with each query taking time $T$. Then we can output a matrix $B$ such that 
 \[\|A-B\|_F \le (1+\eps)\|A-A_k\|_F \]
 where $A_k$ denotes the optimal rank-$k$ approximation to $A$ in Frobenius norm. The runtime is $\tilde{O}(Tk/\eps^{1/3} + nk^{w-1}/\eps^{(w-1)/3})$.
\end{theorem}

Directly appealing to Theorems \ref{thm:musco} and \ref{thm:linear_system} in conjunction with our matrix-vector queries achieves the fastest runtime for computing the top $k$ singular values and solving linear systems for a wide variety of distance matrices that we are aware of. 

\begin{theorem}\label{thm:singular_values_system}
Suppose we have exact matrix-vector query access to a matrix $A$ with each query taking time $T$.
We can compute a $1\pm \eps$ approximation to the $k$ singular values of $A$ in time $\tilde{O}(T/\eps^{1/2} + nk^2/\eps + k^3/\eps^{3/2})$. Furthermore, we can solve linear systems for $A$ with the same guarantees as any iterative method which only uses matrix-vector queries with an multiplicative overhead of $T$.
\end{theorem}

 Finally, we can also perform matrix multiplication faster with distance matrices compared to the general runtime of $n^{w} \approx n^{2.37}$. This follows from the following lemma.

\begin{restatable}{lemma}{matrixmult}
\label{lem:matrixmult}
Suppose $A \in \R^{n \times n}$ admits an exact matrix-vector query algorithm in time $T$. Then for any $B \in \R^{n \times n}$, we can compute $AB$ in time $O(Tm)$.
\end{restatable}

\begin{proof}
We can compute $AB$ by computing the product of $A$ with the $n$ columns of $B$ separately. 
\end{proof}
 
 As a corollary, we obtain faster matrix multiplication for all the family of matrices which we have obtained a fast matrix-vector query for. We state one such corollary for the $\ell_p^p$ case.
 \begin{corollary}
  Let $p \ge 1$ and consider the case that $A_{i,j} = \|x_i - x_j\|_p^p$ for all $i,j$. For any other matrix $B$, we can compute $AB$ in time $O(n^2dp)$.
 \end{corollary}
 
 We can improve upon the above result slightly if we are multiplying two distance matrices for the $p=2$ case.
 
\begin{lemma}\label{lem:matrixmultltwo}
 Consider the case that $A_{i,j} = \|x_i - x_j\|_2^2$ for all $i,j$ and $B_{i,j} = \|y_i - y_j\|_2^2$, i.e., both $A$ and $B$ are $n \times n$ matrices with $f = \ell_2^2$. We can compute $AB$ in time $O(n^2d^{w-2})$.
\end{lemma}

 \begin{proof}
 By decomposing both $A$ and $B$, it suffices to compute the product $XX^TYY^T$ where $X, Y \in \R^{n \times d}$ are the matrices with the points $x_i$ and $y_i$ in the rows respectively. $Z_1 := X^TY \in \R^{d \times d}$ can be computed in $O(nd^2)$ time. Then $Z_2 := XZ_1 \in \R^{n \times d}$ can also be computed in $O(nd^2)$ time. Finally, we need to compute $Z_2 \times Y^T$. This can be done in $O(n^2d^{w-2})$ time by decomposing both $Z_2$ and $Y^T$ into $n/d$ many $d \times d$ square matrices and using the standard matrix multiplication bound on each pair of square matrices. This results in the claimed runtime of $O((n/d)^2 \cdot d^{w}) = O(n^2d^{w-2})$.
 \end{proof}

\section{A Fast Algorithm for Creating \texorpdfstring{$\ell_1$}{L-1} and \texorpdfstring{$\ell_2$}{L-2} Distance Matrices}

We now present a fast algorithm for creating distance matrices which addresses our contribution $(3)$ stated in the introduction. Given a set of $n$ points $x_1 , \ldots, x_n$ in $\R^d$, our goal is to initialize an approximate $n \times n$ distance matrix $B$ for the $\ell_1$ distance which satisfies
\begin{equation}\label{eq:dist_matrix_approx}
    B_{ij} = (1 \pm \eps) \| x_i - x_j\|_1
\end{equation}
for all entries of $B$ where $0 < \eps < 1$ is a precision parameter. The straightforward way to create the \emph{exact} distance matrix takes take $O(n^2d)$ time and by using the stability of Cauchy random variables, we can create $B$ which satisfies \eqref{eq:dist_matrix_approx} in $O(n^2 \log n)$ time for any constant $\eps$. (Note the Johnson-Lindenstrauss lemma implies a similar guarantee for the $\ell_2$ distance matrix). The goal of this section is to improve upon this `baseline' runtime of $O(n^2 \log n)$. The final runtime guarantees of this section will be of the form $O(n^2 \cdot \text{poly}(\log \log n))$.

Our improvement holds in the \textbf{Word-RAM} model of computation. Formally, we assume each memory cell (i.e. word) can hold $O(\log n)$ bits and certain computations on words take $O(1)$ time. The only assumptions we require are the arithmetic operations of adding or subtracting words as well as performing left or right bit shifts on words takes constant time. 

We first present prior work on metric compression of \cite{indyk2017practical} in Section \ref{sec:metric_compression}. Our algorithm description starts from Section \ref{sec:metric_preprocess} which describes our preprocesing step. Section \ref{sec:metric_algo} then presents our key algorithm ideas whose runtime and accuracy are analyzed in Sections \ref{sec:metric_runtime} and \ref{sec:metric_accuracy}.

\subsection{Metric Compression Tree of \texorpdfstring{\cite{indyk2017practical}}{[IRW '17]}}\label{sec:metric_compression}

The starting point of our result is the metric compression tree construction of \cite{indyk2017practical}, whose properties we summarize below. First we introduce some useful definitions. The aspect ration $\Phi$ of $X$ is defined as 
\[\Phi = \frac{\max_{i,j} \|x_i - x_j\|_1}{\min_{i \ne j} \|x_i - x_j\|_1}.\]
Let $\Delta' = \max_{i \in [n]} \|x_1- x_i\|_1$ and $\Delta = 2^{\lceil \log \Delta' \rceil}$.

Theorem $1$ of \cite{indyk2017practical} implies the following result. Given a dataset $X = \{x_1, \ldots, x_n\} \subset \R^d$ with aspect ration $\Phi$, there exists a tree data structure $T$ which allows for the computation of a compressed representation $X$ for the purposes of distance computations. At a high level, $T$ is created by enclosing $X$ in a large enough and appropriately shifted axis-parallel square and then recursively dividing into smaller squares (also called cells) with half the side-length until all points of $X$ are contained in their own cell. The edges of $T$ encode the cell containment relationships. Formally, $T$ has the following properties:
\begin{itemize}
    \item The leaf nodes of $T$ correspond to the points of $X$.
    \item The edges of $T$ are of two types: short edges and long edges which are defined as follows. Short edges have a length $d$ bit vector associated with them whereas long edges have an integer $\le O(\log \Phi)$ associated with them.
    \item Each long edge with associated integer $k$ represents a non-branching path of length $k$ of short edges, all of whose associated length $d$ bit vectors are the $0$ string.
    \item Each node of $T$ (including the nodes that are on paths which are compressed as long edges) have an associated level $- \infty < \ell \le \log(4 \Delta)$. A level $\ell$ of a node $v$ corresponds to an axis-parallel square $G_{\ell}$ of side length $2^\ell$ which contains all axis-parallel squares of child nodes of $v$.
\end{itemize}

The notion of a padded point is important for the metric compression properties of $T$.

\begin{definition}[Padded Point]\label{def:padded_point}
A point $x_i$ is $(\eps, \Lambda, \ell)$-padded, if the grid cell $G_{\ell}$ of side length $2^\ell$ that contains $x_i$ also contains the ball of radius $\rho(\ell)$ centered at $x_i$, where
\[\rho(\ell) = 8 \eps^{-1}2^{\ell - \Lambda}\sqrt{d}. \]
We say that $x_i$ is $(\eps, \Lambda)$-padded in $T$, if it is $(\eps, \Lambda, \ell)$-padded for every level $\ell$.
\end{definition}

The following lemma is proven in \cite{indyk2017practical}. First define 
\begin{equation}\label{def:Lambda}
\Lambda  = \log(16d^{1.5} \log \Phi/(\eps \delta)).
\end{equation}

\begin{lemma}[Lemma 1 in \cite{indyk2017practical}] Consider the construction of $T$ defined formally in Section $3$ of \cite{indyk2017practical}. Every point $x_i$ is $(\eps, \Lambda)$-padded in T with probability $1-\delta$.
\end{lemma}

Now let $x$ be any point in our dataset. We can construct $\widetilde{x} \in \R^d$, called the decompression of $x$, from $T$ with the following procedure: We follow the downward path from the root of $T$ to the leaf associated with $x$ and collect a bit string for every coordinate $d$ of $\widetilde{x}$. When going down a short edge with an associated bit vector $b$, we concatenate the $i$th bit of $b$ to the end of the bit string that we are keeping track of for the $i$th coordinate of $\widetilde{x}$. When going down a long edge, we concatenate with  a number of zeros equalling the integer associated with the long edge. 
Finally, a binary floating point is placed in the resulting bit strings of each coordinate after the bit corresponding to level $0$. The collected bits then correspond to the binary expansion of the coordinates of $\widetilde{x}$. For a more thorough description of the decompression scheme, see Section $3$ of \cite{indyk2017practical}.

The decompression scheme is useful because it allows approximate distance computations.

\begin{lemma}[Lemma 2 in \cite{indyk2017practical}]
If a point $x_i$ is $(\eps, \Lambda)$-padded in $T$, then for every $j \in [n]$,
\[ (1 - \eps) \| \widetilde{x}_i -  \widetilde{x}_j \|_1 \le \|x_i - x_j\|_1 \le (1 + \eps) \| \widetilde{x}_i - \widetilde{x}_j \|_1. \]
\end{lemma}

We now cite the runtime and space required for $T$. The following theorem follows from the results of \cite{indyk2017practical}.

\begin{theorem}\label{thm:tree_properties}
Let $L = \log \Phi + \Lambda$.
$T$ has $O(n \Lambda)$ edges, height $L$, its total size is $O(nd\Lambda + n \log n)$ bits, and its construction time is $O(ndL)$. 
\end{theorem}

We contrast the above gauranttes with the naive representation of $X$ which stores $O(\log n)$ bits of precision for each coordinate and occupies $O(nd \log n)$ bits of space, whereas $T$ occupies roughly $O(nd \log \log n + n \log n)$ bits.

Finally, Theorem $2$ in \cite{indyk2017practical} implies we can create a collection of $O(\log n)$ trees $T$ (by setting $\delta$ to be a small constant in \eqref{def:Lambda}) such that every point in $X$ is padded in at least one tree in the collection.

\subsection{Step 1: Preprocessing Metric Compression Trees}\label{sec:metric_preprocess}

We now describe the preprocessing steps needed for our faster distance matrix compression. Let
\begin{equation}\label{def:w}
  w = \frac{4d\Lambda}{\log n}  
\end{equation}
and recall our setting of $\Lambda$ in \eqref{def:Lambda}. Note that we assume $w$ is an integer which implicitly assumes $4d\Lambda \ge \log n$.

First we describe the preprocessing of $T$. Consider a short edge of $T$ with an associated $d$ length bit string $b$. We break up $b$ into $w$ equal chunks, each containing $d/w$ bits. Consider a single chunk $c$. We pad (an equal number of) $2\Lambda$ many $0$'s after each bit in $c$ so that the total number of bits is $\log n/2$. We then store each padded chunk in one word. We do this for every chunk resulting in $w$ words for each short edge and we do this for all short edges in all the trees.

The second preprocessing step we perform is creating a $O(\sqrt{n}) \times O(\sqrt{n})$ table $A$. The rows and columns of $A$ are indexed by all possible bit strings with $\frac{\log n}2$ bits. The entries of $A$ record evaluations of the function $f(x,y)$ defined as follow: given $x,y$ where $x, y \in \{0,1\}^\frac{\log n}2$, consider the partition of each of them into $d/w$ blocks, each with an equal number of bits ($2\Lambda$ bits per block. Note that $2\Lambda \cdot d/w = (\log n)/2$). Each block defines an integer. Doing so results in $d/w$ integers $x^1, \ldots, x^{d/w}$ derived from $x$ and $w$ integers $y^1, \ldots, y^{d/w}$ derived from $y$. Finally,
\[ f(x,y) = \sum_{i=1}^{d/w} |x^i - y^i|. \]

\subsection{Step 2: Depth-First Search}\label{sec:metric_algo}
We now calculate one row of the distance matrix from point a padded point $x$ to all other points in our dataset. Our main subroutine is a tree search procedure. Its input is a node $v$ in a tree $T$ and it performs a depth-first search on the subtree rooted at $v$ as described in Algorithm \ref{alg:process_tree}. Given an internal vertex $v$, it calculates all the distances between the padded point $x$ to all data points in our dataset which are leaf nodes in the subtree rooted at $v$. A summary of the algorithm follows.

We perform a depth-first search starting at $v$. As we traverse the tree, we keep track of the current embedding of the internal nodes  via collecting bits along the edges of $T$: we append bits when we descend the tree and remove as we move up edges. However, we only keep track of this embedding up to $2\Lambda$ levels below $v$. After that, we continue traversing the tree but don't update the embedding. The reason for this is after $2\Lambda$ levels, the embedding is precise enough for all nodes below with respect to computing the distance to $x$. Towards this end, we also track how many levels below $v$ the tree search is currently at and update this value appropriately. 

The current embedding is tracked using $w$ words. Recall that the bit string of every short edge has been repackaged into $w$ words, each `responsible' for $d/w$ coordinates. Furthermore, in each word on the edge, we have padded $0$'s between the bits of each $d/w$ coordinates. When we need to update the current embedding by incorporating the bits along a short edge $e$, we simply perform a bit shift on each of the $w$ words on $e$ and add it to the $w$ words we are keeping track of. We need to make sure we place the bits `in order.' That is for our tracked embedding, for every $d$ coordinates, the bits on an edge $e$ should precede the bits on the edge directly following $e$ in the tree search. Due to the padding from the preprocessing step, the bit shift implies the bits on the short edges after $e$ will be placed in their appropriate corresponding places in order in the embedding representation.

\begin{algorithm}[H]
\caption{\label{alg:process_tree}DFS in Subtree}
\begin{algorithmic}[1]
\State \textbf{Input:} Metric Compression Tree $T$, node $v$
\Procedure{Search}{ $T,v$}
\State Initialize a global counter $p$ for the number of levels which have been processed. Initially set to $0$ and will be at most $2\Lambda$
\State Initialize $w$ words $t_1, \ldots, t_w$, all initially $0$
\State Initialize a global counter $r$ for the current level which is initially set to the level of $v$ in $T$
\State Perform a depth-first search in the subtree rooted at $v$. Run \texttt{Process-Short-Edge} if a short edge is encountered, \texttt{Process-Long-Edge} if a long edge is encountered, and \texttt{Process-Leaf} when we arrive at a leaf.
\EndProcedure
\end{algorithmic}
\end{algorithm}
While performing the depth-first search, we will encounter both short and long edges. When encountering a short edge, we run the function \texttt{Process-Short-Edge} and similarly, we run \texttt{Process-Long-Edge} when a long edge is encountered. Finally if we arrive at a leaf node, we run \texttt{Process-Leaf}.

\begin{algorithm}[H]
\caption{\label{alg:process_short}Process Short Edge}
\begin{algorithmic}[1]
\State \textbf{Input:} Short edge $e$, number of processed nodes $p$, $t_1, \ldots, t_w$
\Procedure{Process-Short-Edge}{$e, p, t_1,\ldots, t_w$}
\State Let $e_1, \ldots, e_w$ be the $w$ words associated with edge $e$
\State If search is traversing down $e$ and $p < 2\Lambda$,  add $2^{-p}e_i$ to $t_i$ for all $1 \le i \le w$ and increment $p$
\State If search is traversing up $e$ and $p \le 2\Lambda$, subtract $2^{-p}e_i$ from $t_i$ for all $1 \le i \le w$ and decrement $p$
\State Update $r$ to the level of the current node
\EndProcedure
\end{algorithmic}
\end{algorithm}

\begin{algorithm}[H]
\caption{\label{alg:process_long}Process Long Edge}
\begin{algorithmic}[1]
\State \textbf{Input:} Long edge $e$, number of processed nodes $p$
\Procedure{Process-Long-Edge}{$e, p$}
\State   If search is traversing down $e$ and $p < 2\Lambda$, increment $p$
\State If search is traversing up $e$ and $p \le 2\Lambda$, decrement $p$
\State Update $r$ to the level of the current node
\EndProcedure
\end{algorithmic}
\end{algorithm}

When we arrive at a leaf node $y$, we currently have the decompression of $y$ computed from the tree. Note that we only have kept track of the bits after node $v$ (up to limited precision) since all prior bits are the same for $y$ and $x$ since they are in the same subtree. More specifically, we have $w$ words $t_1, \ldots, t_w$. The first word $t_1$ has $2\Lambda$ bits of each of the first $d/w$ coordinates of $y$. For every coordinate, the $2\Lambda$ bits respect the order described in the decompression step in Section \ref{sec:metric_preprocess}. A similar fact is true for the rest of the words $t_i$. Now to calculate the distance between $x$ and $y$, we just have to consider the $2\Lambda$ bits of all $d$ coordinates of $x$ which come after descending down vertex $v$. We then repackage these $2d\Lambda$ total bits into $w$ words in the same format as $y$. Note this preprocessing for $x$ only happens once (at the subtreee level) and can be used for all leaves in the subtree rooted at $v$.

\begin{algorithm}[H]
\caption{\label{alg:process_leaf}Process Leaf}
\begin{algorithmic}[1]
\State \textbf{Input: $t_1, \ldots, t_w$} 
\Procedure{Process-Leaf}{$y, t_1, \ldots, t_w$}
\State Let the point $y$ correspond to the current leaf node
\State Let $s_1,\ldots, s_w$ denote the embedding of $x$ after node $v$, preprocessed to be in the same format as  $t_1, \ldots, t_w$
\State Report $\sum_{i=1}^w A[t_i, s_i]$ as the distance between $x$ and $y$
\EndProcedure
\end{algorithmic}
\end{algorithm}

Finally, the complete algorithm just calls Algorithm \ref{alg:process_tree} on successive parent nodes of $x$. We mark each subtree that has already been processed (at the root node) so that the subtree is only ever visited once. The number of calls to Algorithm \ref{alg:process_tree} is at most the height of the tree, which is bounded by $O(\log \Phi + \Lambda)$. We then repeat this for all points $x$ in our dataset (using the tree which $x$ is padded in) to create the full distance matrix.

\subsection{Runtime Analysis}\label{sec:metric_runtime}
We consider the runtime required to compute the row corresponding to a padded point $x$ in the distance matrix. Multiplying by $n$ results in the total runtime. Consider the tree $T$ in which $x$ is padded in and which we use for the algorithm described in the previous section and recall the properties of $T$ outlined in Theorem \ref{thm:tree_properties}. $T$ has $O(n\Lambda)$ edges, each of which is only visited at most twice in the tree search (going up and down). Thus the time to traverse the tree is $O(n \Lambda)$. There are also at most $O(n \Lambda)$ short edges in $T$. Updating the embedding given by $t_1, \ldots, t_w$ takes $O(w)$ time per edge since it can be done in $O(w)$ total word operations. Long edges don't require this time since they represent $0$ bits; for long edges, we just increment the counter for the current level. Altogether, the total runtime for updating $t_1, \ldots, t_w$ across all calls to Algorithm \ref{alg:process_tree} for the padded point $x$ is $O(n \Lambda w)$. Finally, calculating the distance from $x$ to a fixed point $y$ requires $O(w)$ time since we just index into the array $A$ $w$ times. Thus the total runtime is dominated by $O(n \Lambda w)$. Finally, the total runtime for computing all rows of the distance matrix is 
\[O(n^2 \Lambda w) = O\left(\frac{n^2d \Lambda^2}{\log n}\right) = O\left( \frac{n^2 d}{\log n} \, \log^2\left( \frac{d \log \Phi}{\eps} \right) \right)\]
by setting $\delta$ to be a small constant in \eqref{def:Lambda}.

\subsection{Accuracy Analysis}\label{sec:metric_accuracy}
We now show that the distances we calculate are accurate within a $1\pm \eps$ multiplicative factor. The lemma below shows that if a padded point $x$ and another point $y$ have a sufficiently far away least-common ancestor in $T$, then we can disregard many lower order bits in the decompression computed from $T$ while still guaranteeing accurate distance measurements. The lemma crucially relies on $x$ being padded.

\begin{lemma}
Suppose $x$ is $(\eps, \Lambda)$-padded in $T$. For another point $y$, suppose the least common ancestor of $x$ and $y$ is at level $\ell$. Let $\widetilde{x}$ and $\widetilde{y}$ denote the sketches of $x$ and $y$ produced by $T$. Let $\widetilde{x}'$ be a modified version of $\widetilde{x}$ where for each of the $d$ coordinates, we remove all the bits acquired after level $\ell-2\Lambda$. Similarly define $\widetilde{y}'$. Then
\[ \|\widetilde{x}' - \widetilde{y}'\|_1 = (1 \pm \eps) \|x -y\|_1. \]
\end{lemma}
\begin{proof}
Since $x$ is padded, we know that $\|x-y\|_1 \ge p(\ell-1)$ by Definition \ref{def:padded_point}. On the other hand, if we ignore the bits after level $\ell-2\Lambda$ for every coordinate of $\widetilde{x}$ and $\widetilde{y}$, the additive approximation error in the distance is bounded by a constant factor times 
\[d \cdot \sum_{i=-\infty}^{\ell-2\Lambda-1} 2^i = d \cdot 2^{\ell-2\Lambda}. \]
From our choice of $\Lambda$, we can easily verify that $d \cdot 2^{\ell-2\Lambda} \le\eps p(\ell-1)/2$. Putting everything together and adjusting the value of $\eps$, we have
\[\|\widetilde{x}' - \widetilde{y}'\|_1 = \|\widetilde{x} - \widetilde{y}\|_1 \pm \eps p(\ell-1)/2 =  (1 \pm \eps/2) \|x -y\|_1 \pm \eps p(\ell-1)/2 =  (1 \pm \eps) \|x -y\|_1\]
where we have used the fact that $ \|\widetilde{x} - \widetilde{y}\|_1 =  (1 \pm \eps/2) \|x -y\|_1$ from the guarantees of the compression tree of \cite{indyk2017practical}.
\end{proof}

Putting together our results above along with the Johnson-Lindenstrauss Lemma and Theorem \ref{thm:l2embed} proves the following theorem.

\begin{theorem}\label{thm:metric_compression}
Let $X = \{x_1 , \ldots, x_n\}$ be a dataset of $n$ points in $d$ dimensions with aspect ration $\Phi$. We can calculate a $n \times n$ matrix $B$ such that each $(i,j)$ entry $B_{ij}$ of $B$ satisfies
\[ (1-\eps)\|x_i - x_j\|_1 \le B_{ij} \le   (1+\eps)\|x_i - x_j\|_1\]
 in time 
\[O\left( \frac{n^2 d}{\log n} \, \log^2\left( \frac{d \log \Phi}{\eps} \right) \right).\]
Assuming the aspect ratio is polynomially bounded, we can compute $n \times n$ matrix $B$ such that each $(i,j)$ entry $B_{ij}$ of $B$ satisfies
\[ (1-\eps)\|x_i - x_j\|_2 \le B_{ij} \le   (1+\eps)\|x_i - x_j\|_2\]
with probability $1-1/\text{poly}(n)$. The construction time is
\[O\left( \frac{n^2}{\eps^2} \, \log^2\left( \frac{\log n}{\eps} \right) \right).\]
\end{theorem}

\subsection{A Faster Algorithm for \texorpdfstring{$\ell_{\infty}$}{L-Infinity} Distance Matrix Construction Over Bounded Alphabet}
In this section, we show how to create the $\ell_{\infty}$ distance matrix. Recall from Section \ref{sec:l_infinity_lowerbound} that there exists no $o(n^2)$ time algorithm to compute a matrix-vector query for the $\ell_{\infty}$ distance matrix, assuming SETH, even for $n$ points in $\{0,1,2\}^d$. This suggests that any algorithm for computing a matrix-vector query needs to initialize the distance matrix. However, there is still a gap between a $\Omega(n^2)$ lower bound for matrix-vector queries and the naive $O(n^2d)$ time needed to compute the $\ell_{\infty}$ distance matrix. We make progress towards showing that this gap is not necessary. Our main result is that surprisingly, we can \emph{initialize} the $\ell_{\infty}$ distance matrix in time much faster than the naive $O(n^2d)$ time. 

\begin{theorem}
Given $n$ points over $\{0,1,\ldots, M\}^d$, we can initialize the exact $\ell_{\infty}$ distance matrix in time $O(M^{w-1} n^2 (d \log M)^{w-2})$ where $w$ is the matrix multiplication constant. We can also initialize a  $n \times n$ matrix $B$ such that each $(i,j)$ entry $B_{ij}$ of $B$ satisfies
\[ (1-\eps)\|x_i - x_j\|_{\infty} \le B_{ij} \le   (1+\eps)\|x_i - x_j\|_{\infty}\]
in time $\tilde{O}(\eps^{-1}n^2(dM)^{w-2})$.
\end{theorem}
Thus for $M = O(1)$, which is the setting of the lower bound, we can initialize the distance matrix in time $O(n^2 d^{w-2})$ and thus compute a matrix-vector query in that time as well. 

\begin{proof}
The starting point of the proof is to first design an algorithm which constructs a matrix with $(i,j)$ entry an indicator vector for $\|x-y\|_{\infty} \le i$ or  $\|x-y\|_{\infty} > i$. Given this, we can then sum across all $M$ choices and construct the full distance matrix. Thus it suffices to solve this intermediate task.

Pick a sufficiently large $p$ such that $di^p \le (i+1)^p$. A choice of $p = O(M \log d)$ suffices. Now in the case that $\|x-y\|_{\infty} \le i$, we have $\|x-y\|_p^p \le di^p$ and otherwise, $\|x-y\|_p^p \ge (i+1)^p$. Thus, the matrix $C$ with the $(i,j)$ entry being $\|x_i-x_j\|_p^p$ is able to distinguish the two cases so it suffices to create such a matrix. Now we can write $\|x-y\|_p^p$ as an inner product in $O(pd)$ variables, i.e., it is a gram matrix. Thus computing $C$ can be done by computing a product of $n \times O(pd)$ matrix by a $O(pd) \times n$ matrix, which can be done in 
\[ O\left( \left( \frac{n}{pd} \right)^2 (pd)^{w-2} \right) = O(n^2 (pd)^{w-2}).\]
time by partitioning each matrix into square submatrices of dimension $O(pd)$. Plugging in the bound for $p$ and considering all possible choices of $i$ results in the final runtime bound of $O(M^{w-1} n^2 (d \log M)^{w-2})$, as desired.

Now if we only want to approximate each entry up to a multiplicative $1\pm \eps$ factor, it suffices to only loop over $i$'s which are increasing by powers of $1+c\eps$ for a small constant $c$. This replaces an $O(M)$ factor by an $O(\eps^{-1} \log M)$ factor.
\end{proof}

\end{document}